\documentclass[american,aps,pra,reprint,superscriptaddress,nofootinbib]{revtex4-2}
\usepackage[T1]{fontenc}
\usepackage{color}
\usepackage[utf8]{inputenc}
\usepackage{amsmath}
\usepackage{xcolor}
\usepackage{ragged2e}
\usepackage[singlelinecheck=false]{caption}
\usepackage{subcaption}
\usepackage{bbold}
\DeclareCaptionTextFormat{hyphen}{\justifying\hyphenpenalty=50\exhyphenpenalty=50 #1}
\usepackage{graphicx}
\usepackage{bbm}
\usepackage{babel}
\usepackage{amsmath}
\usepackage{amssymb}
\usepackage{amsthm}
\usepackage{physics}
\usepackage{amssymb}

\usepackage[unicode=true,pdfusetitle,bookmarks=true,bookmarksnumbered=false,bookmarksopen=false,
 breaklinks=false,pdfborder={0 0 0},pdfborderstyle={},backref=false,colorlinks=true]
 {hyperref}
\hypersetup{
 allcolors=magenta}

\hypersetup{
 allcolors=magenta}

\makeatletter
\theoremstyle{plain}
\newtheorem{thm}{\protect\theoremname}
\theoremstyle{plain}
\newtheorem{prop}[thm]{\protect\propositionname}
\theoremstyle{plain}
\newtheorem{definition}{Definition}[section]
\theoremstyle{plain}
\newtheorem{num}[thm]{Numerical conjecture}

\usepackage{times}
\usepackage{txfonts}
\usepackage{braket}
\usepackage{colortbl}

\makeatother

\providecommand{\propositionname}{Proposition}
\providecommand{\theoremname}{Theorem}

\begin{document}
\title{Optimizing entanglement distribution via noisy quantum channels}

\author{Piotr Masajada}
\affiliation{Institute of Fundamental Technological Research, Polish Academy of Sciences, \\ Pawi\'nskiego 5B, 02-106 Warsaw, Poland}

\author{Marco Fellous-Asiani}
\affiliation{Centre of New Technologies,
University of Warsaw, Banacha 2c, 02-097 Warsaw, Poland}

\author{Alexander Streltsov}
\email{streltsov.physics@gmail.com}
\affiliation{Institute of Fundamental Technological Research, Polish Academy of Sciences, \\ Pawi\'nskiego 5B, 02-106 Warsaw, Poland}

\begin{abstract}
    Entanglement distribution is a crucial problem in quantum information science, owing to the essential role that entanglement plays in enabling advanced quantum protocols, including quantum teleportation and quantum cryptography. We investigate strategies for distributing quantum entanglement between two distant parties through noisy quantum channels. Specifically, we compare two configurations: one where the entanglement source is placed at the midpoint of the communication line, and another where it is located at one end. For certain families of qubit channels we show analytically that the midpoint strategy is optimal. Based on extensive numerical analysis, we conjecture that this strategy is generally optimal for all qubit channels. Focusing on the midpoint configuration, we develop semidefinite programming (SDP) techniques to assess whether entanglement can be successfully distributed through the network, and to quantify the amount of entanglement that can be distributed in the process. In many relevant cases the SDP formulation reliably captures the maximal amount of entanglement which can be distributed, if entanglement is quantified using the negativity. We analyze several channel models and demonstrate that, for various combinations of amplitude damping and depolarizing noise, entanglement distribution is only possible with weakly entangled input states. Excessive entanglement in the input state can hinder the channel’s ability to establish entanglement. Our findings have implications for optimizing entanglement distribution in realistic quantum communication networks.
\end{abstract}

\maketitle

\section{Introduction}
Quantum entanglement~\cite{Horodeckientanglementreview} is a fundamental type of correlations between quantum particles, enabling communication protocols that surpass classical limits. As a uniquely quantum phenomenon, entanglement acts as a vital resource for tasks such as quantum teleportation~\cite{Benettteleportofstate}, superdense coding~\cite{Benettsuperdensecoding}, and quantum key distribution~\cite{BennettQKD}. Realizing the full potential of these applications requires the ability to distribute entanglement efficiently and reliably over long distances—an essential step toward scalable quantum networks and the quantum internet.

The aim of this work is to investigate realistic scenarios in which entanglement can be established between two distant parties. While entanglement distribution has been studied in various contexts, here we focus on a setting where a source generates entangled pairs, which are then distributed between two parties, Alice and Bob, through a noisy quantum channel. Specifically, we consider two configurations: one in which the source is placed midway between Alice and Bob, and another in which the source is co-located with one of the two parties. We then compare these setups to determine which yields more favorable outcomes. 

Most existing studies on entanglement distribution have focused on protocols where one party—typically Alice—locally prepares a two-particle entangled state and sends one of the particles to the other party (Bob) through a quantum channel~\cite{Palnotmest,Streltsovdistribution,Krisnandadistributionentanglement,Zuppardentanglementdistribution}, see also Fig.\ref{fig::inkaa}. Interestingly, it has been shown in~\cite{Palnotmest,Streltsovdistribution} that using a maximally entangled state is not always optimal for entanglement distribution. In certain scenarios, a non-maximally entangled state prepared by Alice can lead to a greater amount of entanglement shared between Alice and Bob after transmission. Additionally, protocols involving pre-shared correlations between the parties have been explored in~\cite{Streltsovdistributionentangpre,Chuand}.

The problem of determining whether entanglement can be distributed via a noisy quantum channel is closely related to the fundamental question of whether a given quantum state is entangled—a problem known to be computationally hard in general~\cite{Gurvitsnph}. Similarly, there is no simple characterization of the set of separable (i.e. non-entangled) states~\cite{Horodeckientanglementreview}. However, several sufficient criteria for detecting entanglement exist. One of the most well-known is the Peres-Horodecki criterion~\cite{Perespptcriterion,Horodeckipptcriterion}, which states that if the partial transpose of a state’s density matrix has at least one negative eigenvalue, the state is entangled. In the special cases of two qubit and qubit–qutrit systems, this criterion is also necessary—meaning it fully characterizes entanglement in those dimensions \cite{Horodeckipptcriterion}. However, in higher dimensions, there exist entangled states with a positive partial transpose (PPT)~\cite{HORODECKI1997333,Horodeckipptbound}. 

Many tasks in quantum information science require an optimization over the set of separable states. Because there is no simple description of this set, it is often approximated by the larger set of PPT states. Since every separable state is PPT, the set of separable states is strictly contained within the set of PPT states. Conversely, states that are not PPT—that is, those with a non-positive partial transpose—are referred to as NPT states.

The presence of negative eigenvalues in the partial transpose can be used to quantify entanglement through an entanglement monotone known as negativity~\cite{PhysRevA.58.883,Vidalnegat}, defined as: \begin{equation} 
N(\rho) = \sum_{\lambda_1 < 0} |\lambda_1|. \label{eq::neg} \end{equation} 
Here, $\lambda_1$ are the eigenvalues of the partial transpose of the state $\rho$. In two qubit and qubit–qutrit systems, negativity serves as a valid entanglement measure, meaning that it is nonnegative, and vanishes only on separable states. In this work, we focus on the two qubit case, where negativity simplifies further: it is equal to the absolute value of the single negative eigenvalue of the partial transpose~\cite{Sanperaonenegativeeigenvalue}. 

Leveraging semidefinite programming (SDP) techniques~\cite{Skrzypczykquantumsdp}, we develop methods to estimate whether a given channel configuration can establish long-distance entanglement. SDP is a powerful class of convex optimization problems with broad applications in quantum physics. Owing to the mathematical structure of density operators, many optimization tasks in quantum information theory can be naturally formulated as SDPs. One of the key advantages of SDP is its computational efficiency for numerical optimization. Moreover, each SDP has a corresponding dual problem, which can sometimes be exploited to obtain analytical solutions or tight bounds. Consequently, when a problem admits an SDP formulation, it becomes possible to either solve it analytically or obtain accurate numerical estimates.

In the following sections, we make extensive use of SDP techniques to address a range of problems concerning entanglement distribution. For numerical computations, we utilize the Python package CVXPY~\cite{diamond2016cvxpy,agrawal2018rewriting}, in combination with QuTiP~\cite{JOHANSSONqutippython} and NumPy~\cite{Harrisnumpy}.

This article is organized as follows: In Section~\ref{sec:2}, we investigate the optimal placement of the entanglement source, comparing scenarios where the source is located with one of the parties (Alice or Bob) versus positioned midway between them. We show that placing the source at the midpoint between the two parties generally leads to the most effective entanglement distribution. In Section~\ref{sec:3}, we investigate in more details the midway strategy. Our analysis reveals that, in certain noise regimes, only states with minimal initial entanglement are robust enough to remain entangled after transmission. We further explore scenarios involving different types of quantum channels, identifying conditions under which entanglement can still be generated. We conclude in Section~\ref{sec:Conclusions}.

\section{\label{sec:2} Role of the source location for entanglement distribution}
\begin{figure}
\begin{subfigure}{0.5\textwidth}
    \centering
    \includegraphics[width=\textwidth]{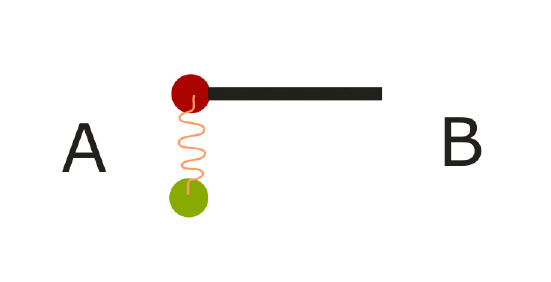}
    \caption{Standard entanglement distribution setup: Alice possesses two entangled particles (green and red circles). She sends one of them to Bob via a -- possibly noisy -- quantum channel (represented by a black line).} 
    \label{fig::inkaa}
\end{subfigure}
\vspace{1cm}
\begin{subfigure}{0.5\textwidth}
    \centering
    \includegraphics[width=\textwidth]{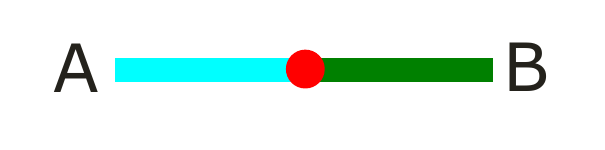}
    \caption{Entanglement distribution protocol considered in this paper: The objective is to establish entanglement between parties \( A \) and \( B \) by placing a source of entangled particles (indicated by the red dot) either midway between Alice and Bob or at one end of the channel. When the source is placed midway, one particle is transmitted to Alice through the blue segment of the channel, and the other to Bob through the green segment. If the blue segment acts as the quantum channel \( \Lambda_1 \), and the green segment as \( \Lambda_2 \), then the overall action of the channel in this configuration is given by \( \Lambda_1 \otimes \Lambda_2 \).}
    \label{fig::inkab}
\end{subfigure}
\begin{subfigure}{0.5\textwidth}
    \centering
    \includegraphics[width=\textwidth]{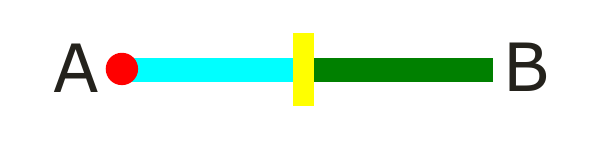}
    \caption{ If the source (red dot) is placed on Alice's side, the action of the channel is given by: \( \openone \otimes \Lambda_2 \Lambda_1. \) We also consider the possibility of inserting a filter (represented by a yellow rectangle) in the middle of the channel. In this case, the action of the overall channel becomes: \( \openone \otimes \Lambda_2 \, \Tilde{\Lambda} \, \Lambda_1.\)   }
    \label{fig::inkac}
\end{subfigure}
\caption{Distribution of entanglement in a different settings.}
\label{fig::strategiesofdistribution}
\end{figure}
We begin this section by recalling the concept of quantum channels. A noisy quantum channel is mathematically described by a completely positive trace-preserving (CPTP) map. Such maps can be represented in several equivalent ways. In this work, we employ two commonly used representations: the Kraus operator formalism \cite{Nielsenquantumcomputation} and the Choi–Jamiołkowski isomorphism \cite{CHOIch}. For a comprehensive overview of quantum channels, see \cite{gyongyosipropertieschannel}.

Suppose the goal is to distribute entanglement between two parties, Alice and Bob. The simplest and most widely studied approach involves Alice locally preparing a pair of entangled particles and sending one of them to Bob via a quantum channel~\cite{Palnotmest,Streltsovdistribution}. In this work, we will consider a more general strategy, see Fig. \ref{fig::strategiesofdistribution}. We assume the presence of an entangled source that emits pairs of entangled particles. Our goal is to determine the optimal position of the source for establishing entanglement between Alice and Bob. In particular, we compare two scenarios: placing the source at the midpoint between the parties and placing it at one end of the channel. When the source is located at the midpoint (Fig.~\ref{fig::inkab}), one particle travels through the left segment of the channel, and the other through the right. Denoting the action of the left and right segments by $\Lambda_1$ and $\Lambda_2$, respectively, the effective action of the full channel is given by:
\begin{equation}
\Lambda_\mathrm{midway} = \Lambda_1 \otimes \Lambda_2.
\label{eq::parallel}
\end{equation}
If the source is instead co-located with Alice (see Fig.~\ref{fig::inkac}), the action of the channel is described by
\begin{equation}
\Lambda_{\mathrm{Alice}} = \openone\otimes\Lambda_2 \circ \Lambda_1.
\label{eq::serialc}
\end{equation}
Analogously, co-locating the source with Bob gives the channel 
\begin{equation}
    \Lambda_{\mathrm{Bob}} = \Lambda_1 \circ \Lambda_2 \otimes \openone. \label{eq::serialr}
\end{equation}

We also allow the insertion of a filter between the left and right segments of the channel. In this case, the overall channel takes the form, (see Fig.~\ref{fig::inkac}):
\begin{equation}
\Lambda_f = \openone\otimes\Lambda_2 \circ\Tilde{\Lambda}\circ\Lambda_1.
\label{eq::serialf}
\end{equation}
$\Tilde{\Lambda}$ in the equation above is a completely-positive (CP) map which plays the role of the filter. 

Before proceeding, we introduce the following definitions.
\begin{definition}
    A bipartite quantum channel $\Lambda_1 \otimes \Lambda_2$ is said to be entanglement-annihilating (EA) if, for every bipartite input state $\ket{\psi}$, the output state $\Lambda_1 \otimes \Lambda_2(\ket{\psi}\bra{\psi})$ is separable.
\end{definition}

\begin{definition}{\label{de::eb}}
    A quantum channel $\Lambda$ is called entanglement-breaking (EB) if, for every bipartite input state $\ket{\psi}$, the output state $(\openone \otimes \Lambda)(\ket{\psi}\bra{\psi})$ is separable.
\end{definition}

Entanglement-annihilating channels have been investigated in~\cite{Filippovea, MoravcikovaEAandEBchannels, Filippovbipartiteeachannels}, while entanglement-breaking channels were studied in~\cite{horodecki2003entanglement, RuskaiEB, HolevoEBIN}.

Assume we are comparing two strategies for distributing entanglement between two parties. Strategy X uses the channel $\Lambda_X$, and strategy Y uses the channel $\Lambda_Y$. In both cases, the channels have a tensor product structure as described in equations~\eqref{eq::parallel} and~\eqref{eq::serialc}. We say that strategy Y is better than strategy X if the fact that $\Lambda_X$ preserves entanglement for some input state (i.e., it is not entanglement-annihilating) implies that $\Lambda_Y$ also preserves entanglement for some input state. In the following, we determine the best strategy for placing the entangled source.

It is known that the channels $\openone \otimes \Lambda$ and $\Lambda \otimes \openone$ are entanglement-annihilating (EA) if and only if $\Lambda$ is entanglement-breaking (EB)~\cite{Filippovea}. As a consequence, the strategy described in Eq.~\eqref{eq::serialc} is not EA if the composed channel $\Lambda_2 \circ \Lambda_1$ is not EB. Similarly, the strategy in Eq.~\eqref{eq::serialr} is not EA if $\Lambda_1 \circ \Lambda_2$ is not EB. 
Therefore, if we demonstrate that placing the source midway is better than placing it at Alice’s end, we effectively show that for some set of channels $\Lambda_1$ and $\Lambda_2$, the implication
\begin{equation}
    \Lambda_2 \circ \Lambda_1 \text{ is not EB} \Rightarrow \Lambda_1 \otimes \Lambda_2 \text{ is not EA}
    \label{eq:EBimpliesEA}
\end{equation}
holds. Additionally, we assume that when the source is at one edge of the channel, it is positioned at the left end of the communication channel. This is without loss of generality. In different terms, for the case the entanglement source is at one edge of the communication channel, we focus our analysis on the channel of the form $\openone \otimes \Lambda_2 \circ \Lambda_1$ for entanglement distribution, as indicated on Fig \ref{fig::inkac}.

\begin{definition}
    \label{def:Pauli}
    A Pauli channel is a channel which acts as a probabilistic mixture of Pauli matrices (including the identity). Pauli matrices are given by:
    \begin{equation}
\begin{aligned}
\sigma_x &= \begin{pmatrix}
0 & 1 \\
1 & 0
\end{pmatrix}, \quad
\sigma_y = \begin{pmatrix}
0 & -i \\
i & 0
\end{pmatrix}, \quad
\sigma_z = \begin{pmatrix}
1 & 0 \\
0 & -1
\end{pmatrix}
\end{aligned}
\label{eq::pauli}
\end{equation}
The Kraus operators of the Pauli channel are: 
$K^{\text{Pauli}}_0=\sqrt{p_0}\openone,$ $K^{\text{Pauli}}_1=\sqrt{p_1}\sigma_x,$ 
$K^{\text{Pauli}}_2=\sqrt{p_2}\sigma_y$ and $K^{\text{Pauli}}_3=\sqrt{p_3}\sigma_z.$ Each 
$p_i\in [0,1]$ with additional condition $p_0+p_1+p_2+p_3=1.$ 
\end{definition}

\begin{definition}
    A transposition of a quantum channel $\Lambda$ is a completely positive map having Kraus operators $K_i^{T},$ where $K_i$ are the Kraus operators of the channel $\Lambda.$ The exponent $T$ on $K_i^T$ denotes the transposition of the Kraus operator in some basis. We denote the resulting channel by $\Lambda^{T}.$ 
\end{definition}

Note that the transposition of a quantum channel does not depend on the choice of the Kraus representation of a quantum channel.

\begin{prop}\label{pr::unitarybettermidway}
Let $\Lambda_1$ be a unital quantum channel acting on a qubit, $\Lambda_2$ an arbitrary quantum qubit channel, and $\widetilde{\Lambda}$ a completely positive map (not necessarily trace preserving) that we call a filter. If the composed channel $\openone \otimes \Lambda_2 \circ \widetilde{\Lambda} \circ \Lambda_1$ preserves entanglement for some input state (i.e., it is not EA), then the product channel $\Lambda_1 \otimes \Lambda_2$ also preserves entanglement for some input state.
\end{prop}
\begin{proof}
    Because $\Lambda_1$ is a unital channel, we can write it as: 
    \begin{align}
        \Lambda_1 = \mathcal{U} \circ \mathcal{P} \circ \mathcal{V},
    \end{align}
    where $\mathcal{P}$ is a Pauli channel, and $\mathcal{U},\mathcal{V}$ are unitary channels~\cite{li2023unital}. 

      Assume that $\Lambda_2 \circ \widetilde{\Lambda} \circ \Lambda_1$ is not EB. Consider the state $\ket{\beta}=(\openone \otimes V^{\dagger}) \ket{\phi^+}$, where $\ket{\phi^+}=\frac{1}{\sqrt{2}}\left(\ket{00}+\ket{11}\right)$ is the maximally entangled 2-qubit state, and $V$ is a unitary matrix associated to $\mathcal{V}$, i.e., we have $\mathcal{V}(\rho)=V \rho V^{\dagger}$ for any qubit density matrix $\rho$. The state $ \openone \otimes \Lambda_2 \circ \widetilde{\Lambda} \circ \Lambda_1  \left(\ketbra{\beta}{\beta} \right)$ is entangled, since a channel is EB if and only if its Choi state is separable \cite{horodecki2003entanglement}. Moreover:
    \begin{align}
        \left( \openone \otimes \Lambda_2 \circ \widetilde{\Lambda} \circ \Lambda_1 \right) \left(\ketbra{\beta}{\beta} \right)&=\left(\openone \otimes \Lambda_2 \circ \widetilde{\Lambda} \circ \mathcal{U} \circ \mathcal{P} \circ \mathcal{V} \right) \left(\ketbra{\beta}{\beta} \right) \notag \\
        &=\left(\openone  \otimes \Lambda_2 \circ \widetilde{\Lambda} \circ \mathcal{U} \circ \mathcal{P} \right) \left(\ketbra{\phi^+}{\phi^+} \right) \notag \\
        &=\left(\mathcal{P}^{T_z}  \otimes \Lambda_2 \circ \widetilde{\Lambda} \circ \mathcal{U} \right) \left(\ketbra{\phi^+}{\phi^+} \right),
    \end{align}
    where $\mathcal{P}^{T_z}$ corresponds to the Pauli channel, which has its Kraus operators $(I,\sigma_x,\sigma_y,\sigma_z)$ transposed in the eigenbasis of $\sigma_z$ (we use the fact $\openone \otimes M \ket{\phi^+}=M^{T_z} \otimes \openone \ket{\phi^+}$ for any matrix $M$) Notice that $\mathcal{P}^{T_z}=\mathcal{P}$. Then, we have, for $\rho=\openone \otimes \widetilde{\Lambda} \circ \mathcal{U} \ketbra{\phi^+}{\phi^+}$: 
    \begin{align}
        \left( \openone \otimes \Lambda_2 \circ \widetilde{\Lambda} \circ \Lambda_1 \right) \left(\ketbra{\beta}{\beta} \right)&=\left(\mathcal{P} \otimes \Lambda_2 \right) \left(\rho \right).
    \end{align}
    Finally, defining $\rho' \equiv (\mathcal{V}^{\dagger} \otimes \openone)(\rho)$, we have:
    \begin{align}
        \left( \mathcal{U} \otimes \Lambda_2 \circ \widetilde{\Lambda} \circ \Lambda_1 \right) \left(\ketbra{\beta}{\beta} \right)&=\left(\Lambda_1 \otimes \Lambda_2 \right) \left(\rho' \right).
    \end{align}
    In the last equality, we applied a closure 
    relation $\mathcal{V} \circ \mathcal{V}^{\dagger}$ on 
    the right of $\mathcal{P}$, and we applied a local 
    unitary $\mathcal{U} \otimes \openone$ on the 
    left. Because $\openone \otimes \Lambda_2 \circ \widetilde{\Lambda} \circ \Lambda_1  \left(\ketbra{\beta}{\beta} \right)$ is entangled, and $\mathcal{U} \otimes \openone$ does not change entanglement property, $\left(\Lambda_1 \otimes \Lambda_2 \right) \left(\rho' \right)$ is entangled. It proves our property.
\end{proof}

Proposition~\ref{pr::unitarybettermidway} shows that if Alice's part of the channel is unital, then placing the source midway between Alice and Bob leads to a better entanglement distribution strategy, even when a filtering operation \( \widetilde{\Lambda} \) is allowed in the middle of the channel (i.e., at the location of the yellow bar in Fig.~\ref{fig::inkac}) when the source is placed at the edge.
We now proceed to prove a related and useful result.

\begin{prop}{\label{pr::transpositionsimulation}}

    If we can decompose $\Lambda_1^{T}$ as $\Lambda_1^T = \mathcal{F} \circ \Lambda_1 \circ \mathcal{E}$, where $T$ denotes the channel transposition in some basis, $\mathcal{F}$ is a positive map and $\mathcal{E}$ is a CP map, then the following implication holds: If $\openone \otimes \Lambda_2 \circ \Lambda_1$ preserves entanglement for some input state (i.e., it is not EA), then $\Lambda_1 \otimes \Lambda_2$ preserves entanglement for some input state.
\end{prop}
\begin{proof}

Assume that $\openone \otimes \Lambda_2 \circ \Lambda_1$ is not EA. Then the state, $\openone \otimes \Lambda_2 \circ \Lambda_1 ( \ketbra{\phi^+}{\phi^+})$ is entangled, since a channel is entanglement-breaking if and only if its Choi state is separable \cite{horodecki2003entanglement}. Therefore,
\begin{align}
    \openone \otimes \Lambda_2 \circ \Lambda_1 ( \ketbra{\phi^+}{\phi^+}) &= \Lambda_1^T \otimes \Lambda_2 (\ketbra{\phi^+}{\phi^+}) \notag \\
    &= \mathcal{F} \circ \Lambda_1 \circ \mathcal{E} \otimes \Lambda_2 ( \ketbra{\phi^+}{\phi^+}) \notag \\
    &= c \times \mathcal{F} \circ \Lambda_1 \otimes \Lambda_2 (\rho) 
    \text{ is also entangled.}\label{eq:putting_lambda_1_transpose_prepostprocessed} 
    \end{align}
    In the formula above, we substituted:
    \begin{align}
    \rho &\equiv (\mathcal{E} \otimes \openone) (\ketbra{\phi^+}{\phi^+})/c,  \label{eq::m}\\
    c &\equiv \Tr((\mathcal{E} \otimes \openone) (\ketbra{\phi^+}{\phi^+})).
\end{align}
We begin by excluding the case $c=0.$ Note that $\Tr((\mathcal{E} \otimes \openone) (\ketbra{\phi^+}{\phi^+}))=\Tr(\mathcal{E}(\openone))=c.$ $c=0$ implies that   $\mathcal{E} ( \openone)=0,$ since $\mathcal{E}$ is CP. In that case, the decomposition $\Lambda_1^T = \mathcal{F} \circ \Lambda_1 \circ \mathcal{E}$ is not possible, and thus this scenario must be excluded. 
In Eq. \eqref{eq::m}, the state $\rho$ is a valid quantum state due to the complete positivity of $\mathcal{E}$. What remains to be shown is that the map $\mathcal{F} \otimes \openone$ cannot generate entanglement from a separable state. If this is the case, then $\Lambda_1 \otimes \Lambda_2(\rho)$ is necessarily entangled, which is what we want to show. 

 Assume that we can decompose $\Lambda_1 \otimes \Lambda_2(\rho)$ as: $\widetilde{\rho} \equiv \Lambda_1 \otimes \Lambda_2(\rho)=\sum_i p_i \sigma_i \otimes \delta_i$ for some density matrices $\sigma_i$, $\delta_i$ (i.e. the state $\Lambda_1 \otimes \Lambda_2(\rho)$ is separable). We have:
\begin{align}
    c (\mathcal{F} \otimes \openone)(\widetilde{\rho})&=\sum_i c p_i \Tr(\mathcal{F}(\sigma_i)) \times \mathcal{F}(\sigma_i)/\Tr\left(\mathcal{F}\left(\sigma_i\right)\right) \otimes \delta_i.
\end{align}
Thus, the state $c \times \mathcal{F} \circ \Lambda_1 \otimes \Lambda_2 (\rho)$ is separable if $\Lambda_1 \otimes \Lambda_2 (\rho)$ is separable. Additionally, each ($\mathcal{F}(\sigma_i)/\Tr(\mathcal{F}(\sigma_i))$ is a positive semi definite matrix of trace $1$, hence a valid density matrix)\footnote{Note that by taking the trace of the whole quantity we have $\sum_i c p_i \Tr(\mathcal{F}(\sigma_i))=1$.}. From Eq. \eqref{eq:putting_lambda_1_transpose_prepostprocessed} we obtain that if $\openone \otimes \Lambda_2 \circ \Lambda_1 ( \ketbra{\phi^+}{\phi^+})$ is entangled, then $c \times \mathcal{F} \circ \Lambda_1 \otimes \Lambda_2 (\rho)$ is entangled, and consequently $\Lambda_1 \otimes \Lambda_2 (\rho)$ must be entangled. This observation finishes the proof.
\end{proof}

We now apply this proposition to demonstrate that, for qubit channels with Kraus rank at most 3, placing the source midway between Alice and Bob is almost always (in a precise sense defined in Proposition \ref{prr:rank3_simulation}) the optimal strategy. 

\begin{prop}
Assume that $\Lambda$ is a quantum channel with Kraus rank at most 3, i.e., it can be written as $\Lambda(\rho) = \sum_{i=1}^3 K_i \rho K_i^{\dagger}$. Then, we can almost always find two matrices $A$ and $B$ such that for any $i \in [1,3]$,
\begin{align}
    A K_i B = K_i^T.
\end{align}
This identity implies that:
\begin{align}
    \Lambda^T(\rho)=\mathcal{F} \circ \Lambda \circ \mathcal{E},
\end{align}
where $\mathcal{E}$ and $\mathcal{F}$ are Kraus rank-1 CP maps. By "almost always", we mean that it is always true apart from a zero-measure set in the space of Kraus operators. More precisely, because the matrices $A$ and $B$ involve fractions with denominators that can vanish for specific choices for the family of Kraus operators $\{K_i\}_{i=1}^3$, there is always a solution, apart from a zero-measure set.
\label{prr:rank3_simulation}
\end{prop}
\begin{proof}
    
    To demonstrate this, we parametrized each Kraus operator using four parameters and used Mathematica to derive explicit expressions for the matrices $A$ and $B$ as functions of the family $\{K_i\}.$ Due to the length and complexity of these expressions, we do not include them in this article.

In the course of this derivation, we found that some coefficients in the expressions for $A$ and $B$ involve denominators that may vanish for specific choices of the Kraus operators. As a result, our construction fails in these exceptional cases. However, these problematic choices form a set of measure zero in the space of all possible Kraus operator families.
\end{proof}

We now state Proposition \ref{prr:middle_better_analytical} which is one of the main results of this section. 
\begin{prop}
If $\Lambda_1$ has Kraus rank at most $3$, then the following proposition holds almost always (i.e. it is always true apart from a set of measure zero within the space of all possible Kraus operators for $\Lambda_1$: see Proposition \ref{prr:rank3_simulation}). 

If the channel $\openone \otimes \Lambda_2 \circ \Lambda_1$ preserves entanglement for some input state (i.e. it is not EA), then the product channel $\Lambda_1 \otimes \Lambda_2$ also preserves entanglement for some input state.

\label{prr:middle_better_analytical}
\end{prop}
\begin{proof}
    From Proposition \ref{prr:rank3_simulation}, we know that \begin{align}
    \Lambda_1^T(\rho)=\mathcal{F} \circ \Lambda_1 \circ \mathcal{E},
\end{align}
where $\mathcal{E}$ and $\mathcal{F}$ are rank-1 CP maps, is almost always true. We then apply Proposition \ref{pr::transpositionsimulation} to conclude the proof.
\end{proof}
We have shown that if the Kraus rank of the channel $\Lambda_1$ is at most 3, then it is almost always better to place the source midway between Alice and Bob. In the following, we provide numerical evidence that, in the qubit-qubit scenario, placing the source midway is always optimal.
To establish this, we determine when the channel in Eq. \eqref{eq::serialc} always produces a separable output. This occurs precisely when the composed channel $\Lambda_2 \circ \Lambda_1$ is entanglement-breaking (EB). The Choi matrix of the channel $\Lambda_2 \circ \Lambda_1$ is given by \cite{KhatriPQCT}:

\begin{equation}
    C_{12}^{s}=\Tr_{A'B}(\openone^{A}\otimes2\times\phi^{+A'B}\otimes\openone^{B'} C_1^{AA'}\otimes C_{2}^{BB'}),
    \label{eq::choiserialconnection}
\end{equation}
where $C_1^{AA'}$ and $C_2^{BB'}$ denote the Choi matrix of the channels $\Lambda_1$ and $\Lambda_2$, respectively, and $\phi^{+}$ is the density matrix of the maximally entangled state. A quantum channel is entanglement-breaking (EB) if and only if its Choi matrix is separable~\cite{horodecki2003entanglement}.
Therefore, for qubit channels $\Lambda_1$ and $\Lambda_2$, we can determine whether the composed channel $\Lambda_2 \circ \Lambda_1$ is EB by checking whether the Choi matrix given in Eq.~\eqref{eq::choiserialconnection} has a positive partial transpose (PPT). Concretely, this amounts to evaluating the minimal eigenvalue of the partially transposed Choi matrix $C_{12}^{s T_B}$. If this eigenvalue is positive, then the Choi matrix is PPT, and hence $\Lambda_2 \circ \Lambda_1$ is entanglement-breaking.

We have established a simple criterion for determining whether the channel corresponding to the source placed at the edge can produce an entangled output. In the following, we present an analogous result for the case where the source is placed midway between the parties. We formulate the following proposition:

\begin{prop}\label{pr::sdprelaxationchannel}
    Assume that we are given a fixed channel of the form $\Lambda_1\otimes\Lambda_2$
where $\Lambda_1$ and $\Lambda_2$ are quantum channels. Let us denote by $\Tilde{N}(\rho)$ the minimal eigenvalue of the partial transpose of a bipartite state $\rho$. Then the following inequality holds: 
\begin{equation}
    \min_{\ket{\psi}^{AB}}\Tilde{N}(\Lambda_1\otimes\Lambda_2(\ket{\psi}^{AB}\bra{\psi}^{AB}))\geq\min_{X^{T_{A'B'}}\geq0}\Tr\left(X^{ABA'B'}C^{T_B'}\right),
    \label{eq::propositionsdprelaxation}
\end{equation}
where $C$ denotes the Choi matrix of the channel $\Lambda_1\otimes\Lambda_2$. 
\end{prop}

\begin{proof}
The following chain of expressions proves the proposition.
    \begin{equation}
    \begin{split}
        &\min_{\ket{\psi}^{AB}}\Tilde{N}(\Lambda_1\otimes\Lambda_2(\ket{\psi}^{AB}\bra{\psi}^{AB}))=\\
        &=\min_{\ket{\psi},\ket{\varphi}}\bra{\varphi}\Tr_{AB}\left((\ket{\psi}\bra{\psi}^{AB})^T\otimes\openone^{A'B'}C\right)^{T_{B'}}\ket{\varphi}=\\
        &=\min_{\ket{\psi},\ket{\varphi}}\bra{\varphi}\Tr_{AB}\left(\ket{\psi}\bra{\psi}^{AB}\otimes\openone^{A'B'}C^{T_{B'}}\right)\ket{\varphi}=\\
        &=\min_{\ket{\psi},\ket{\varphi}}\Tr_{AB}\left(\ket{\psi}\bra{\psi}^{AB}\otimes\ket{\varphi}\bra{\varphi}^{A'B'}C^{T_{B'}}\right)=\\
        &=\min_{X^{ABA'B'}\in S^{AB|A'B'}}\Tr_{AB}\left(X^{ABA'B'}C^{T_{B'}}\right)\geq\\
        &\min_{X^{T_{A'B'}}\geq0}\Tr\left(X^{ABA'B'}C^{T_{B'}}\right).
    \end{split}
    \label{eq::propositionproof}
\end{equation}
In the second equality, we bring the partial transpose operation inside the minimization. In the final step, we use the fact that the set of separable states $S$ is a subset of the set of PPT states~\cite{Horodeckipptcriterion,Perespptcriterion}. The notation $C^{T_{B'}}$ indicates that the partial transpose is taken with respect to the $B'$ subsystem of the Choi matrix, which acts on the composite system $ABA'B'$.
\end{proof}

Proposition \ref{pr::sdprelaxationchannel} provides a lower bound for the minimal eigenvalue of the partial transpose of the channel's output. If the right-hand side of the relation \eqref{eq::propositionsdprelaxation} is non-negative, it follows that the smallest eigenvalue of the output's partial transpose is also non-negative. In the qubit-qubit and qubit-qutrit cases, this implies that the channel is entanglement-annihilating \cite{Horodeckipptcriterion,Perespptcriterion}.

We introduce two minimization criteria to determine which source placement configuration yields better entanglement distribution. Based on these criteria, we formulate the following conjecture.

\begin{num}
    Consider two-qubit system. Let us denote minimal eigenvalue of partial transpose of \eqref{eq::choiserialconnection} by $\lambda_{min}(C_{12}^{s T_B}).$ We have:
   \begin{equation}
       \frac{\lambda_{min}(C_{12}^{s T_B})}{2}\geq \min_{X^{T_{A'B'}}\geq0}\Tr\left(X^{ABA'B'}C^{T_B'}\right).
       \label{eq::strategiesinequality}
   \end{equation}
\end{num}
    
We tested 200{,}000 randomly generated quantum channels and did not observe a single case in which the above conjecture failed.

 Moreover, in scenarios where the bound in equation~\eqref{eq::propositionsdprelaxation} is saturated, which, as we show later in this paper, occurs frequently, equation~\eqref{eq::strategiesinequality} can be used to establish that if \( \Lambda_1 \circ \Lambda_2 \) is not entanglement-breaking (EB), then \( \Lambda_1 \otimes \Lambda_2 \) is not entanglement-annihilating (EA). The reasoning is as follows: if \( \Lambda_1 \circ \Lambda_2 \) is not EB, then the left-hand side of equation~\eqref{eq::strategiesinequality} is negative. Consequently, the right-hand side must also be negative, which implies that
\(
\min_{\ket{\psi}^{AB}} \widetilde{N}\left((\Lambda_1 \otimes \Lambda_2)\left(\ket{\psi}^{AB}\bra{\psi}\right)\right) < 0.
\)
This, in turn, indicates that the output of \( \Lambda_1 \otimes \Lambda_2 \) is entangled for some input state, meaning the channel is not EA.
These findings support the conclusion that there is no scenario in which placing the source at the end of the channel yields entangled output, while placing it in the middle always results in a separable state.

\section{\label{sec:3} Source placed midway between Alice and Bob}
 Previously, we showed that placing the source midway between the parties typically leads to a better strategy. We now turn our attention to analyzing this configuration in more details. In the following sections, we demonstrate that, in the qubit–qubit setting, the relation given in Eq.~\eqref{eq::propositionsdprelaxation} often holds as an equality. Furthermore, in this case, the partial transpose has at most one negative eigenvalue \cite{Sanperaonenegativeeigenvalue}, allowing us to upper bound the output negativity using Proposition \ref{pr::sdprelaxationchannel}. 

Assume that the source is placed midway between Alice and Bob, resulting in a channel of the form given by Eq.~\eqref{eq::parallel}. The channels $\Lambda_1$ and $\Lambda_2$ are quantum channels that may depend on multiple parameters. Our objective is to determine the range of these parameters for which the channel is capable of distributing entanglement. In other words, we aim to identify the parameter regimes where the channel is not entanglement-annihilating, that is, where there exists at least one input state that yields entangled output.
Furthermore, using Eq.~\eqref{eq::propositionsdprelaxation}, we will compute the maximal negativity of the channel output (i.e., the absolute value of the minimal negative eigenvalue of the partial transpose of the output state), and identify the optimal input state—defined as the state that achieves this maximal negativity. The procedure is as follows: using Proposition~\ref{pr::sdprelaxationchannel}, we first compute a lower bound on the minimal eigenvalue of the partial transpose. We then attempt to identify an input state that attains this bound. If such a state exists, the bound is tight, and we thereby determine both the minimal eigenvalue and the state that achieves it.

\subsection{Depolarizing - amplitude damping}{\label{s::da}}

Assume that in Eq. \eqref{eq::parallel}, $\Lambda_1$ is a depolarizing channel, denoted by $\Lambda^{\text{Depol}}_p$, and $\Lambda_2$ is an amplitude damping channel, denoted by $\Lambda^{\text{AD}}_\gamma$.
The Kraus operators of the amplitude damping channel are given by:
\begin{equation}
    K^{\text{AD}}_0=\left(
\begin{array}{cc}
 1 & 0 \\
 0 & \sqrt{1-\gamma}  \\
\end{array}
\right),
\end{equation}
\begin{equation}
    K^{\text{AD}}_1=\left(
\begin{array}{cc}
 0 & \sqrt{\gamma}  \\
 0 & 0 \\
\end{array}
\right),
\label{def:kraus_amplitude_damping}
\end{equation}
where $\gamma\in[0,1]$ is the damping parameter determining 
the strength of the noise. The depolarizing 
channel $
\Lambda^{\text{Depol}}_p$ is given by the following four Kraus operators: $K^{\text{Depol}}_0=\sqrt{1-\frac{3}{4}p_s}\openone,$ $K^{\text{Depol}}_1=\sqrt{\frac{1}{4}p_s}\sigma_x,$ $K^{\text{Depol}}_2=\sqrt{\frac{1}{4}p_s}\sigma_y,$ $K^{\text{Depol}}_3=\sqrt{\frac{1}{4}p_s}\sigma_z.$ Where $\sigma_{x}, \sigma_{y}, \sigma_{z}$ are the Pauli matrices defined in Definition \ref{def:Pauli} and $p_s\in[0,1]$ is depolarizing parameter. The depolarizing channel 
is EB if and only if $p_s\geq\frac{2}{3}$ \cite{KhatriPQCT}. 

We can express the input state of the channel using the 
Schmidt decomposition:
\begin{equation}
    \ket{\psi}=\sqrt{c}\ket{a}\otimes\ket{b}+\sqrt{1-c}\ket{a^{\bot}}\otimes\ket{b^{\bot}},
    \label{eq::schmidt}
\end{equation}
The parameter $c,$ which range is in $[0,0.5]$ corresponds to the geometric entanglement of the state \cite{Weigeometricentanglement}. 
 For any pair of orthonormal states we can always find a diagonal unitary $V$ such that
\begin{equation}
    \begin{split}
        &\ket{s}=V\ket{b}=\sqrt{s_1}\ket{0}+\sqrt{1-s_1}\ket{1},\\
        &\ket{s^{\bot}}=V\ket{b^{\bot}}=-\sqrt{1-s_1}\ket{0}+\sqrt{s_1}\ket{1},
        \label{eq::sbot}
    \end{split}
\end{equation}
 where $s_1$ is a real number in the interval $[0,1].$ The amplitude damping channel commutes with diagonal unitaries.  We can exploit this symmetry and change the initial state using Eq.~\eqref{eq::sbot}. Such a change will not alter the output entanglement. Additionally, the depolarizing channel commutes with all unitaries. Therefore, we can also do a change of basis. Hence, the following state will have the same output entanglement as the state in Eq.~\eqref{eq::schmidt}:
\begin{equation}
    \ket{\varphi}=\sqrt{c}\ket{0}\otimes\ket{s}+\sqrt{1-c}\ket{1}\otimes\ket{s^{\bot}}.
    \label{eq::psidad}
\end{equation}

In the following, we determine the range of channel parameters for which the channel is not entanglement-annihilating (EA). Recall that the depolarizing channel becomes entanglement-breaking when $p_s \geq \frac{2}{3}.$ Thus, in this regime, the output is always separable, regardless of the input state. Similarly, the amplitude damping channel is entanglement-breaking for $\gamma=1.$ Let $\rho^{\text{Depol,AD}}_f$ denote the output state of the channel. 
\begin{equation}   
\rho^{\text{Depol,AD}}_f=\Lambda^{\text{Depol}}_p\otimes\Lambda^{\text{AD}}_\gamma(\ket{\varphi}\bra{\varphi}),
    \label{eq::outdad}
\end{equation}
where $\ket{\varphi}$ is given by Eq.~\eqref{eq::psidad}. In the two-qubit scenario, a state is entangled if and only if it is NPT~\cite{Perespptcriterion,Horodeckipptcriterion}. Moreover, if a state is NPT, its partial transpose has exactly one negative eigenvalue \cite{Sanperaonenegativeeigenvalue}.
 Hence, the determinant of the partial transpose of a two-qubit entangled state is negative \cite{Augusiakdeterminant}.  Therefore, to determine whether a given channel is not EA, we need to find the minimal value of $\det ((\rho^{\text{Depol,AD}}_f)^{T_B}),$ where the minimization is performed over the set of all input states, $\ket{\varphi}$ (more precisely, this set is obtained from $\ket{\varphi}$ in Eq. \eqref{eq::psidad}, where $c$ and $s_1$ are varied such that $c \in [0,0.5]$ and $s_1 \in [0,1]$). If this determinant is non-negative for all inputs, then the partial transpose is always positive semi-definite, and the output state cannot be entangled. Consequently, the channel must be entanglement-annihilating.
 
 In what follows, we will find the minimum of $\det ((\rho^{\text{Depol,AD}}_f)^{T_B}).$ 
Solving $\frac{d}{ds_1}\det((\rho^{\text{Depol,AD}}_f)^{T_B})=0$ for $s_1$ we obtain three solutions:
\begin{equation}
    \begin{split}
        &s_1= \frac{c}{2 c-1}\\
        &s_1= \frac{2 c^2 \gamma-\sqrt{(1-2 c)^2 (c-1) c \gamma}-c \gamma}{4 c^2 \gamma-4
   c \gamma+\gamma}\\
   &s_1= \frac{2 c^2 \gamma+\sqrt{(1-2 c)^2 (c-1) c \gamma}-c \gamma}{4 c^2 \gamma-4
   c \gamma+\gamma}.
    \end{split}
    \label{eq::detsol}
\end{equation}
For $c=0$ all solutions are 0 and the output is always separable -- a separable input cannot create an entangled output.
For $0< c\leq \frac{1}{2},$ the expression under the square root in Eq. \eqref{eq::detsol} is negative. Consequently, solutions 2 and 3 in Eq. \eqref{eq::detsol} are complex. Since $s_1$ is a real number in the $[0,1]$ interval, we can discard these solutions. The first solution is negative, meaning that it is not in the $[0,1]$ interval. Therefore, the equation $\frac{d}{ds_1}\det((\rho^{\text{Depol,AD}}_f)^{T_B})=0$ has no physically valid solution and the extremal values occur at $s_1=0$ and $s_1=1.$ Let us denote the value of this determinant for $s_1=0$ by $d_{s_1=0}$ and for $s_1=1$ by $d_{s_1=1}.$ There is the following relation:
\begin{equation}
    d_{s_1=0}-d_{s_1=1}=\frac{1}{16} (2 c-1) (\gamma -1)^2 \gamma  (2 (c-1) c (\gamma
   -1)+\gamma ) \left(p_s-2\right){}^2 p_s^2.
   \label{eq::ds0mds1}
\end{equation}
Since $c\in[0,0.5],$ the above expression is always non-positive. This implies that the determinant reaches its minimum at $s_1=0.$ Consequently, if the channel is not EA, there will exist an output entangled state with $s_1=0$ and some $c \in [0,0.5]$. We now investigate this scenario. We have:  
\begin{align} \label{eq::ds0}
d_{s_1=0}= & \frac{1}{16}c^{2}(\gamma-1)^{2}\left(p_{s}-2\right){}^{2}\left(c\left((\gamma-3)p_{s}+2\right)+3p_{s}-2\right)\nonumber \\
 & \left(c\left(\gamma p_{s}+p_{s}-2\right)-p_{s}+2\right).
\end{align}
The last bracket can be rewritten as:
\begin{equation}
    c \left(\gamma  p_s+p_s-2\right)-p_s+2=c \gamma  p_s+(1-c) \left(2-p_s\right).
\end{equation}
It is clear that it is positive. Assume that in second to last bracket we have fixed $p_s$ and $\gamma.$ If 
\begin{equation}
    c<\frac{2-3 p_s}{\gamma  p_s-3 p_s+2},
    \label{eq::smacon}
\end{equation}
this bracket is negative. Since it is the only potentially negative term in Eq. \eqref{eq::ds0}, it provides a condition for the existence of an entangled output state. The denominator of Eq. \eqref{eq::smacon} remains positive for $p_s<\frac{2}{3}.$ As $p_s$ approaches $
\frac{2}{3}$ the right-hand side of Eq. \eqref{eq::smacon} can become arbitrarily small, meaning that, for $p_s$ close to $\frac{2}{3}$, only states with extremely weak initial entanglement can produce an entangled output. 

It is a known fact that maximally entangled states are not always optimal for entanglement distribution tasks \cite{Palnotmest}. In particular, Ref.~\cite{Streltsovdistribution} shows that, when using a single-qubit amplitude damping channel and measuring entanglement via negativity, states with minimal initial entanglement can sometimes outperform maximally entangled ones. Our findings extend these results by demonstrating that, in the bipartite setting, the ability to generate entanglement at the output may rely exclusively on states with minimal initial entanglement. We present the above result as a proposition.
\begin{prop}
    Assume that we have a channel $\Lambda^{\text{Depol}}_p\otimes\Lambda^{\text{AD}}_\gamma,$ where $\Lambda^{\text{Depol}}_p$ is a depolarizing channel of associated probability $p_s$ and $\Lambda^{\text{AD}}_\gamma$ is an amplitude damping channel of damping parameter $\gamma$ (see the text around \eqref{def:kraus_amplitude_damping} for the precise definitions). For fixed $p_s, \gamma$, there exists an input state so that we obtain an entangled output after the application of $\Lambda^{\text{Depol}}_p\otimes\Lambda^{\text{AD}}_\gamma$ only if the entanglement of the initial state satisfies $c<\frac{2-3 p_s}{\gamma  p_s-3 p_s+2}.$
    \label{prr:depol_amplitude_damping} In particular, for $p_s$ approaching $2/3$ from below, this result implies that only weakly entangled input state can lead to entangled output.
\end{prop}
Previously, we used the determinant to identify the range of initial entanglement values that can lead to entangled output states. While the determinant is a useful tool in the qubit-qubit scenario, it is not a proper entanglement measure \cite{Augusiakdeterminant}. In what follows, we aim to determine the input state that maximizes the output negativity.
We observe that the determinant is minimized for input states with $s_1 = 0$. However, such states do not necessarily maximize the negativity of the output state. To find the optimal input state, we first used Proposition~\ref{pr::sdprelaxationchannel} to obtain a lower bound on the minimal eigenvalue of the partial transpose of the output state. We then found an input state—characterized by specific values of $s_1$ and $c$—for which the output state, after application of the channel, has the minimum eigenvalue of its partial transpose reaching this lower bound. Since this lower bound is reached for this particular state, it implies that the identified input state is the one leading to the lowest eigenvalue for the partial transpose of the output state. Therefore, it is the input state that maximizes the negativity of the output state.

 We performed the computations as follows: First, we fixed the channel parameters $p_s$ and $\gamma$. Then, we discretized the parameter space of the initial state by sampling $s_1$ and $c$ in steps of $0.01$, starting from $0$ until 1, thereby generating a $100 \times 100$ grid of candidate states. For each point in this grid, we computed the channel output and evaluated the minimal eigenvalue of its partial transpose.
Next, we compared these eigenvalues with the lower bound obtained via Proposition~\ref{pr::sdprelaxationchannel} computed through SDP. We identified a pair $s_1,$ $c$ for which the minimal eigenvalue of the output state's partial transpose matched the lower bound.

 After this, we varied the channel parameters and repeated the procedure described in the previous paragraph. Specifically, we explored a grid of channel parameter values defined by $p_s = 0.01n$ for $n = 0, 1, \dots, 67$ and $\gamma = 0.01n$ for $n = 0, 1, \dots, 99$. Our results indicate that the maximum output negativity is consistently achieved for $s_1 = 0$, while the optimal value of $c$ depends on the specific channel parameters (see Fig.~\ref{fig:depadipl}). This figure shows how the output negativity changes with the initial entanglement for different strength of damping parameter. Fig. \ref{fig:contourminimalentanglement} in turn, shows how the initial entanglement of the optimal input states varies with the noise parameters. Notably, as $p_s$ approaches $\frac{2}{3}$, states with minimal initial entanglement become optimal (see Fig.~\ref{fig:contourminimalentanglement}). In summary, we obtain maximum negativity for the input state of the form:
\begin{equation}
    \sqrt{c}\ket{01}+\sqrt{1-c}\ket{10},
    \label{eq::dami}
\end{equation}
for some $c\in [0,0.5].$ 
\begin{figure}
    \includegraphics[width=0.88\linewidth]{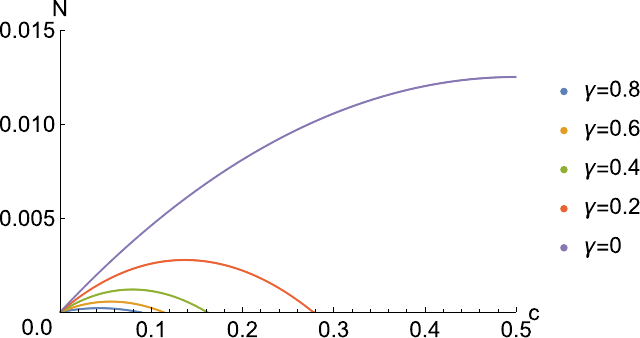}
    \caption{Plot of the maximal negativity of the depolarizing - amplitude damping channel as a function of initial entanglement for different amplitude damping parameters. The depolarizing parameter is constant and equal to $p_s=0.65.$}
    \label{fig:depadipl}
\end{figure}
\begin{figure}
    \includegraphics[width=0.9\linewidth]{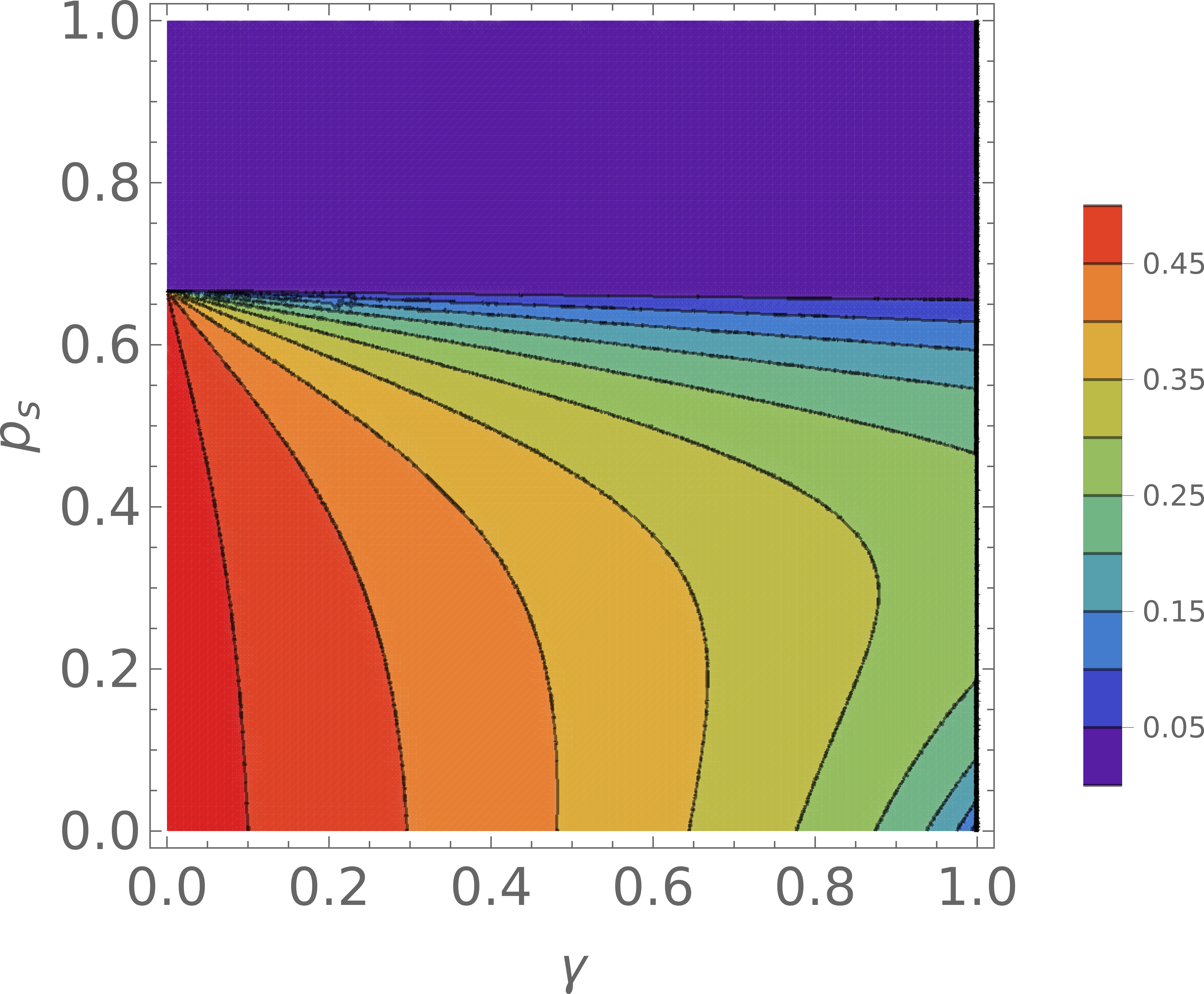}
    \caption{
    Initial geometric entanglement \( c \) of the input state \( \ket{\varphi} \) (defined in~\eqref{eq::psidad}) that maximizes the negativity of the output state \( \rho^{\text{Depol,AD}}_f \) (given in~\eqref{eq::outdad}). The vertical axis represents the depolarizing parameter \( p_s \), while the horizontal axis corresponds to the amplitude damping parameter \( \gamma \).
For strong depolarization (i.e., as \( p_s \) approaches \( \frac{2}{3} \)), input states with very low initial entanglement yield the most entangled output states, as measured by the negativity. In this regime, weakly entangled input states outperform maximally entangled ones in terms of output entanglement. For example, at fixed \( \gamma \), increasing \( p_s \) towards \( \frac{2}{3} \) requires the input state to have decreasing initial entanglement in order to optimize the output.
The purple region in the upper part of the plot indicates that the output state is always separable for \( p_s \geq \frac{2}{3} \), as the depolarizing channel becomes entanglement-breaking in this regime.
    }
    \label{fig:contourminimalentanglement}
\end{figure}

\subsection{Depolarizing-phase flip}
Now let us consider the depolarizing-phase flip channel. The Kraus operators of the phase flip channel are given by: $K_0^{\text{Phase-Flip}}=\sqrt{r}\openone$ and $K_1^{\text{Phase-Flip}}=\sqrt{1-r}\sigma_z,$ where $\sigma_z$ is the Pauli $z$ matrix. As before, $p_s\in[0,1]$ denotes the depolarizing parameter and $r\in[0,1]$ denotes the phase flip parameter.
 Since the Kraus operators of the phase flip channel are diagonal, the channel commutes with diagonal unitaries. As a result, we can express the initial state using Eq. \eqref{eq::psidad}.
  Let $\rho_f^{\text{Depol,Phase-Flip}}$ denote the output of the channel.
\begin{equation}
    \rho_f^{\text{Depol,Phase-Flip}}=\Lambda^{\text{Depol}}_p\otimes\Lambda^{\text{Phase-Flip}}_r(\ket{\varphi}\bra{\varphi}),
    \label{eq::outdpf}
\end{equation}
where $\Lambda^{\text{Depol}}_p$ stands for the depolarizing channel and $\Lambda^{\text{Phase-Flip}}_r$ stands for the phase flip channel. 
Using Wolfram Mathematica we analytically found $\det((\rho_f^{\text{Depol,Phase-Flip}})^{T_B})$. As the resulting expression is lengthy, we choose not to display it explicitly here. We then solve $\frac{d}{ds_1}\det((\rho_f^{\text{Depol,Phase-Flip}})^{T_B})=0$ for $s_1$ and we obtain three different solutions:
\begin{equation}
    \begin{split}
        &s_1= \frac{1}{2}\\
        &s_1= \frac{1}{2} \left(1+\frac{\sqrt{(1-r) r \left((1-c) c (1-2 r)^2+(1-r)
   r\right)}}{(1-2c) (1-r) r}\right)\\
   &s_1= \frac{1}{2} \left(1-\frac{\sqrt{(1-r) r \left((1-c) c (1-2 r)^2+(1-r)
   r\right)}}{(1-2c) (1-r) r}\right).
    \end{split}
    \label{eq::depffsol}
\end{equation}
Let us focus on the expression:
\begin{equation}
    F(c) \equiv \frac{\sqrt{(1-r) r \left((1-c) c (1-2 r)^2+(1-r)
   r\right)}}{(1-2c) (1-r) r}.
\end{equation}
 We have $F(c) \geq F(0)$ (we recall $c \in [0,0.5]$). Therefore:
\begin{equation}
    \frac{\sqrt{(1-r) r \left((1-c) c (1-2 r)^2+(1-r)
   r\right)}}{(1-2c) (1-r) r}\geq\frac{\sqrt{(1-r)^2 r^2}}{(1-r) r}=1.
   \label{eq::dffsol}
\end{equation}
Using this we see immediately that the second solution in Eq.~\eqref{eq::depffsol} is greater than 1, while the third solution is smaller than 0. Thus, given the fact $s \in [0,1]$, $\det((\rho_f^{\text{Depol,Phase-Flip}})^{T_B})$ can reach its minimum at $s_1=0,$ $s_1=\frac{1}{2}$ or $s_1=1.$ $\det((\rho_f^{\text{Depol,Phase-Flip}})^{T_B})$ takes the same value for $s_1=0$ and $s_1=1,$ meaning that we need to check only $s_1=0$ and $s_1=1/2$.

Let $d_{s_1=1}$ be the value of $\det((\rho_f^{\text{Depol,Phase-Flip}})^{T_B})$ for $s_1=1$ and $d_{s_1=1/2}$ be its value for $s_1=1/2$. We have: 
\begin{equation}
\begin{split}
    &d_{s_1=1}-d_{s_1=1/2}=-\frac{1}{16} (1-2 c)^2 (1-r) r \\&\left((1-2 c)^2 (1-r) r+2 (1-c)
   c\right) 
   \left(p_s-2\right){}^2 p_s^2.
\end{split}   
\end{equation}
All the terms in brackets in the above equation are positive, making the right-hand side of this equation negative. Consequently, $d_{s_1=1}\leq d_{s_1=1/2}.$ Therefore, the determinant $\det((\rho_f^{\text{Depol,Phase-Flip}})^{T_B})$ reaches its minimum at $s_1=0$ and $s_1=1.$ We have:
\begin{equation}
\begin{split}
   & d_{s_1=1}=d_{s_1=0}=-\frac{1}{16} (c-1)^2 c^2 \left(p_s-2\right){}^2 \left(4 r
   \left(p_s-1\right)-3 p_s+2\right)\\
   &\left(4 r
   \left(p_s-1\right)-p_s+2\right).
\end{split}
\end{equation}
Only the last two brackets can be negative, considering only these two brackets we obtain the following condition for $d_{s_1=1}=0:$
\begin{equation}
    \begin{split}
        &r= \frac{3 p_s-2}{4 \left(p_s-1\right)}\\
        &r= \frac{p_s-2}{4 \left(p_s-1\right)}.
        \label{depfffsol}
    \end{split}
\end{equation}
These two solutions define two lines in the $r$-$p_s$ plane. In the region between them, the output is separable for all input states. Outside this region, there exist initial states that produce entangled output (see Fig.~\ref{fig:depffp}). Note that our result does not depend on the initial geometric entanglement, $c$. Thus, even input states with an arbitrary small amount of entanglement can produce an entangled output.

Numerical computations show that the maximal output negativity is always achieved with a maximally entangled input state, i.e., $c=1/2$. The computations were performed in a similar manner as what is described in the paragraph following Proposition \ref{prr:depol_amplitude_damping}. 

\begin{figure}
    \includegraphics[width=0.9\linewidth]{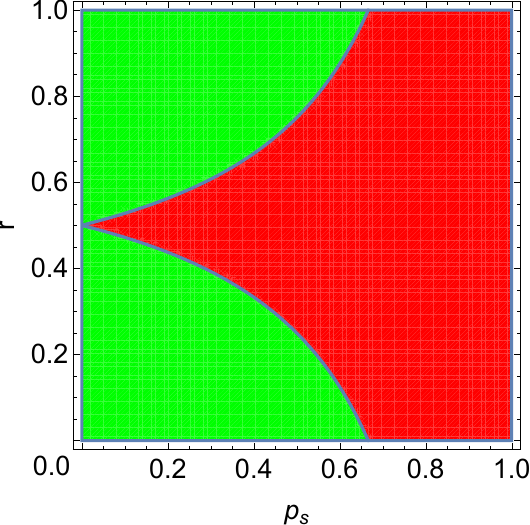}
    \caption{Schematic plot of the entanglement of the output of the channel in Eq.~\eqref{eq::parallel} for the depolarizing-phase flip channel. The red colored region is where the output of the channel is separable, while the green colored region corresponds to the existence of an entangled output.}
    \label{fig:depffp}
\end{figure}
\subsection{Amplitude damping-phase flip}
Now consider the amplitude damping - phase flip channel. Let $\gamma$ be a damping parameter and $r$ the phase flip parameter. Because both channels commute with diagonal unitaries, we can express the input state of the channel as (see Section \ref{s::da}):
\begin{equation}
    \ket{\psi}=\sqrt{c}\ket{s}\otimes\ket{v}+\sqrt{1-c}\ket{s^{\bot}}\otimes\ket{v^{\bot}},
    \label{eq::psiadf}
\end{equation}
where $\ket{s}=\sqrt{s_1}\ket{0}+\sqrt{1-s_1}\ket{1}$ and $\ket{s^{\bot}}=-\sqrt{1-s_1}\ket{0}+\sqrt{s_1}\ket{1}.$ Similarly, $\ket{v}=\sqrt{s_2}\ket{0}+\sqrt{1-s_2}\ket{1}$ and $\ket{v^{\bot}}=-\sqrt{1-s_2}\ket{0}+\sqrt{s_2}\ket{1}.$ The determinant of the partial transpose of the channel's output is a complicated expression. However, we can identify the range in which the output may be entangled by specifying particular values of the input state parameters. For $s_1=0$ and $s_2=0,$ the determinant is equal to: 
\begin{equation}
    (1-c)^2 c^2 (1-\gamma ) (\gamma -1) \left(2 r-1\right)^{2}.
\end{equation}
 The equation above is negative, unless $r=\frac{1}{2}$ or $\gamma=1.$ However, the phase flip channel is EB for $r=\frac{1}{2},$ while the amplitude damping channel is EB for $\gamma=1.$ Our channel cannot produce entangled output if one of the subchannels is EB. Thus, the channel is EA if and only if at least one of the subchannels is EB. 

We used $s_1=0$ and $s_2=0$ to determine the range of parameters for which the channel is entanglement-annihilating. However, this choice does not yield the maximal output negativity. Numerical computations—performed in a manner similar to that described in the paragraph following Proposition \ref{prr:depol_amplitude_damping}—indicate that the maximal output negativity is achieved for $s_1=0$ and $s_2=0.5,$ when $\gamma<0.8.$ For these parameters, the initial state is 
\begin{equation}
    \frac{1}{2}\left(\sqrt{1-c}\ket{00}-\sqrt{1-c}\ket{01}+\sqrt{c}\ket{10}+\sqrt{c}\ket{11}\right).
    \label{eq::apfl}
\end{equation}
Numerical computations further show that for $\gamma>0.8$ this state does not match with the minimal eigenvalue obtained from the SDP. This discrepancy arises because, in this case, the SDP does not provide the correct minimal eigenvalue. We can demonstrate this using the following procedure. First, we extract the matrix $X^{ABA'B'}$ which minimizes Eq.~\eqref{eq::propositionproof}. Then, following the method from \cite{Ganardilocalpurity}, we show that $X^{ABA'B'}$ is entangled in bipartition $AB|A'B'.$ This explains why we have a strict inequality in the Proposition~\ref{pr::sdprelaxationchannel} in this scenario. 
\begin{table*}[t] 
    \centering
    \begin{tabular}{c l c r c}
    \hline
    \hline
        Channel & Parameters & Optimal state & Optimal $c$ & Equality in Proposition \ref{pr::sdprelaxationchannel}\\
        \hline
        Depol-AD & $p_s,\gamma$ &Eq. \eqref{eq::dami} & analytical & all cases\\
        
        Depol-pf & $p_s,r$ &$\ket{\phi^{+}}$ & $\frac{1}{2}$ & all cases\\
        
        AD-pf  &  $\gamma,r$& Eq. \eqref{eq::apfl} & numerical & $\gamma<0.8$\\
        
        AD-bf  & $\gamma,r$ &Eq. \eqref{eq::apfl} & numerical & $\gamma<0.8$\\
        
        GAD-GAD & $n,\gamma$&Eq. \eqref{eq::GAD} & numerical & all cases\\
        \hline       
        \hline
    \end{tabular}
    \caption{The table summarizes various quantum channels along with the corresponding input states that maximize the output negativity. The final column indicates whether equality holds in Proposition~\ref{pr::sdprelaxationchannel}, and, if so, specifies the parameter values for which this occurs.
In the column labeled ``optimal \( c \)'', there are three possible cases. If an explicit value is given, it means that the optimal \( c \) is independent of the channel parameters. If the entry is marked as \emph{analytical}, it indicates that an analytical expression for the optimal \( c \) has been derived. If the entry is marked as \emph{numerical}, it means that no closed-form expression was found, and the value of \( c \) was obtained through numerical optimization.
The table was constructed using a method similar to the one described after Proposition~\ref{prr:depol_amplitude_damping}.
The following abbreviations are used: Depol – depolarizing channel, AD – amplitude damping, pf – phase flip, bf – bit flip, GAD – generalized amplitude damping.}
    \label{tab::sdp}
\end{table*}
\subsection{Pauli Channel}

In this section, we consider a channel composed of the tensor product of two different Pauli channels. The input state can always be expressed in its Schmidt decomposition form~\eqref{eq::schmidt}. Numerical simulations over 10,000 samples indicate that the minimal eigenvalue of the output state's partial transpose is consistently achieved when the input state has maximal geometric entanglement, corresponding to \( c = \frac{1}{2} \). This property is remarkable as we find it to be true for \textit{any} Pauli channel. Our numerical approach was performed in a similar manner as the one described after Proposition \ref{prr:depol_amplitude_damping}. Specifically, we discretized the range of each input state parameter and searched for an output state that reaches the bound from Proposition~\ref{pr::sdprelaxationchannel}. However, the orthonormal basis states appearing in the Schmidt decomposition~\eqref{eq::schmidt} may vary depending on the specific parameters of the channel. Consequently, the minimal eigenvalue is attained for a maximally entangled input state represented in a local basis, which depends on the considered Pauli channels. 

\subsection{Generalized amplitude damping channels}
In this section, we consider a noise modeled by local generalized amplitude damping (GAD) channels. These channels can describe a wide range of quantum noise effects. For instance, they are used to model decoherence in superconducting-circuit-based quantum computing \cite{ChirolliGADquantumcomputing} and to characterize losses in linear optical systems under low-temperature background noise \cite{ZouGADliop}. For a comprehensive review of GAD channels, see \cite{KhartiGADchannel}.

The Kraus operators of the GAD channel are given by:
\begin{equation}
    K_0^{GAD}=\sqrt{1-n}\left(
\begin{array}{cc}
 1 & 0 \\
 0 & \sqrt{1-\gamma}  \\
\end{array}
\right),
\label{eq:K0GAD}
\end{equation}
\begin{equation}
    K_1^{GAD}=\sqrt{1-n}\left(
\begin{array}{cc}
 0 & \sqrt{\gamma}  \\
 0 & 0 \\
\end{array}
\right),
\end{equation}
\begin{equation}
    K_2^{GAD}=\sqrt{n}\left(
\begin{array}{cc}
 \sqrt{1-\gamma}  & 0 \\
 0 & 1 \\
\end{array}
\right),
\end{equation}
\begin{equation}
    K_3^{GAD}=\sqrt{n}\left(
\begin{array}{cc}
 0 & 0 \\
 \sqrt{\gamma}  & 0 \\
\end{array}
\right),
\end{equation}
where $\gamma\in[0,1]$ is a damping parameter and $n\in[0,1]$ is an additional parameter of the channel which can have various interpretations depending on the situation the channel is describing. For example, if the channel describes interaction with a thermal environment, $n$ denotes the probability that the environment is not in its ground state \cite{MyattGAD}. For $n=0$ GAD reduces to the standard amplitude damping channel. Since GAD channel commutes with diagonal unitaries, we can write the initial state as Eq. \eqref{eq::psiadf}. We consider the scenario in which both GAD channels are identical. It was previously studied in \cite{Filippovea} and we will extend these results. We will find the range of parameters where GAD channel is EA. A criterion for which GAD channel is EB was established in \cite{LamiGAD,Filippovea}.

If both GAD channels are identical, the system is characterized by two channel parameters. To identify the input state that maximizes the output negativity, we employ a procedure similar to that described below Proposition~\ref{prr:depol_amplitude_damping}. Specifically, we discretize the range of each input state parameter and search for an output state that reaches the bound from Proposition~\ref{pr::sdprelaxationchannel}.
Numerical computations indicate that, in this scenario, the state 
\begin{equation}
    \sqrt{c}\ket{00} + \sqrt{1-c}\ket{11},
    \label{eq::GAD}
\end{equation}
for some \( c \in [0,1] \), minimizes the minimal eigenvalue of the output's partial transpose. However, numerical results suggest that, to determine whether such a channel preserves entanglement for some input state, it suffices to verify whether it preserves the entanglement of the maximally entangled state \( \ket{\Psi^{-}} = (\ket{01} - \ket{10})/\sqrt{2} \). We formulate this observation as a conjecture.
\begin{num}
    Assume that we have a channel in the form $\Lambda\otimes\Lambda,$ where $\Lambda$ is some GAD channel. Then if $\Lambda\otimes\Lambda(\ket{\psi}\bra{\psi})$ is entangled for some $\ket{\psi},$ $\Lambda\otimes\Lambda(\ket{\Psi^{-}}\bra{\Psi^{-}})$ is also entangled. 
\end{num}

In other words, if we want to check if the channel made of two identical copies of GAD is EA it is enough to check the separability of its application on $\ket{\Psi^{-}}$. To motivate this conjecture we do the following. First, we solve $\lambda_{min}(\Lambda\otimes\Lambda(\ket{\Psi^{-}}\bra{\Psi^{-}})^{T_B})=0$ for $\gamma.$ Here $\lambda_{min}(\rho)$ denotes the minimal eigenvalue of $\rho.$ We obtain the following result:
\begin{equation}
    \gamma=\frac{(n-1) \sqrt{\frac{-8 \left(\sqrt{2}-1\right) (n-1) n-2
   \sqrt{2}+3}{(n-1)^2 n^2}} n+\sqrt{2}-1}{4 (n-1) n}.
   \label{eq::gadgam}
\end{equation}
This implies that if $\gamma$ is given by the relation above and $n$ is a free channel parameter, the minimal eigenvalue of the partial transpose of $\Lambda\otimes\Lambda(\ket{\Psi^{-}}\bra{\Psi^{-}})$ is 0. Next, we sampled $n$ in steps of 0.0001 starting from 0 until $1$. $\gamma$ was given by Eq.~\eqref{eq::gadgam}. Then, we computed SDP from the Proposition \ref{pr::sdprelaxationchannel}. The result of the SDP was consistently zero up to $10^{-8}$. Therefore, the set of channel parameters for which $\Lambda \otimes \Lambda(\ket{\Psi^{-}}\bra{\Psi^{-}})$ transitions from being entangled to separable coincides with the point at which the channel becomes entanglement-annihilating (EA) (this last point being indicated by the fact the result of the SDP vanishes). Hence, the two notions align, as stated in our conjecture. Note that if we decrease the $\gamma$ parameter to a value below the right-hand side of \eqref{eq::gadgam}, $\Lambda\otimes\Lambda(\ket{\Psi^{-}}\bra{\Psi^{-}})$ becomes entangled, meaning the channel is no longer EA. Fig. \ref{fig::gadloc} shows regions where such channel is EA. The results are summarized in Table \ref{tab::sdp}.

\begin{figure}
    \centering
    \includegraphics[width=0.9\linewidth]{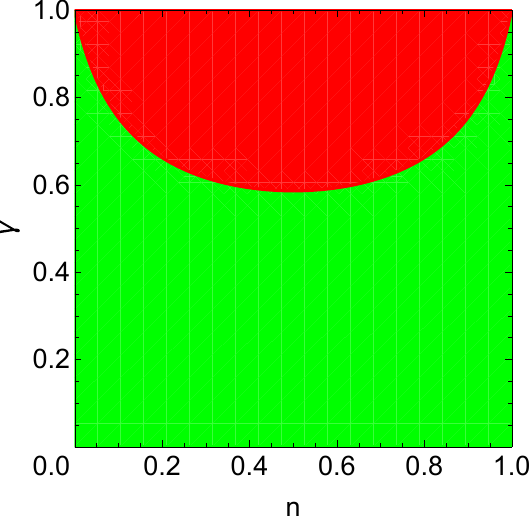}
    \caption{Schematic plot illustrating the entanglement properties of the output state of the channel~\eqref{eq::parallel}, where both subchannels are identical and correspond to the GAD channel defined around equation~\eqref{eq:K0GAD}. The red region indicates parameter values for which the channel output is separable for all input states. In contrast, the green region corresponds to parameter regimes where there exists at least one input state that yields an entangled output.
To determine whether the channel is entanglement-annihilating (EA) in this setting, it suffices to verify whether the singlet state \(\ket{\Psi^{-}}\) produces an entangled output. }
    \label{fig::gadloc}
\end{figure}

\section{\label{sec:Conclusions} Conclusions}
 We explored two strategies for distributing entanglement between two parties, focusing on the distribution of a two-qubit entangled state. For clarity, we refer to them as the first and second strategy. In the first strategy, the entanglement source is placed midway along the communication line (see Fig.\ref{fig::inkab}); in the second strategy, the source is located at one end of the channel (see Fig.\ref{fig::inkac}). We analytically showed that the first strategy outperforms the second when the channel co-located with the source is either unital or has Kraus rank strictly less than 4. More precisely, in the case of Kraus rank less than 4, the advantage holds for almost all channels (see Proposition~\ref{prr:middle_better_analytical} for the precise meaning of "almost all"). For channels with Kraus rank 4, our analytical methods do not yield a definitive conclusion. To address this case, we performed a numerical analysis, from which we conjecture that placing the entanglement source in the middle of the communication line, that is, following the first strategy, is always the better option.

Subsequently, we focused our analysis on the first strategy. Our goal was to find the range of channel parameters for which entanglement can be distributed, for various typical channels. We first demonstrated that the SDP introduced in Proposition~\ref{pr::sdprelaxationchannel} provides a lower bound on the minimal eigenvalue of the output state's partial transpose (the output state being the state received at the end of the entanglement distribution protocol). This bound also often happens to be tight in the sense that it is achieved for specific input states produced by the entanglement source. The interest of this bound is that, when tight, it provides a valuable tool for identifying the parameter ranges that enable entanglement distribution by the channels under study, thanks to the PPT criterion. However, in the case of amplitude damping combined with phase flip noise, we could show that the bound is not always tight for specific values of the channel parameters. 

Finally, we came to the surprizing result that for a combination of depolarizing and amplitude damping noise, in some parameter range for the depolarizing channel, only input states with arbitrarily small initial entanglement can produce entangled outputs. In different terms, strongly entangled input states become separable after the channel's action. It means that it is detrimental to the entanglement distribution task if the source produces input states containing too much entanglement.

\section*{Acknowledgments}
We thank Tulja Varun Kondra for discussion. This work was supported by the National Science Centre Poland (Grant No. 2022/46/E/ST2/00115) and within the QuantERA II Programme (Grant No. 2021/03/Y/ST2/00178, acronym ExTRaQT) that has received funding from the European Union's Horizon 2020 research and innovation programme under Grant Agreement No. 101017733, and the European Union's Horizon Europe Research and Innovation programme under the Marie Sklodowska-Curie Actions \& Support to Experts programme (MSCA Postdoctoral Fellowships) - grant agreement No. 101108284.

\bibliography{biblio}

\begin{thebibliography}{42}%
\makeatletter
\providecommand \@ifxundefined [1]{%
 \@ifx{#1\undefined}
}%
\providecommand \@ifnum [1]{%
 \ifnum #1\expandafter \@firstoftwo
 \else \expandafter \@secondoftwo
 \fi
}%
\providecommand \@ifx [1]{%
 \ifx #1\expandafter \@firstoftwo
 \else \expandafter \@secondoftwo
 \fi
}%
\providecommand \natexlab [1]{#1}%
\providecommand \enquote  [1]{``#1''}%
\providecommand \bibnamefont  [1]{#1}%
\providecommand \bibfnamefont [1]{#1}%
\providecommand \citenamefont [1]{#1}%
\providecommand \href@noop [0]{\@secondoftwo}%
\providecommand \href [0]{\begingroup \@sanitize@url \@href}%
\providecommand \@href[1]{\@@startlink{#1}\@@href}%
\providecommand \@@href[1]{\endgroup#1\@@endlink}%
\providecommand \@sanitize@url [0]{\catcode `\\12\catcode `\$12\catcode `\&12\catcode `\#12\catcode `\^12\catcode `\_12\catcode `\%12\relax}%
\providecommand \@@startlink[1]{}%
\providecommand \@@endlink[0]{}%
\providecommand \url  [0]{\begingroup\@sanitize@url \@url }%
\providecommand \@url [1]{\endgroup\@href {#1}{\urlprefix }}%
\providecommand \urlprefix  [0]{URL }%
\providecommand \Eprint [0]{\href }%
\providecommand \doibase [0]{https://doi.org/}%
\providecommand \selectlanguage [0]{\@gobble}%
\providecommand \bibinfo  [0]{\@secondoftwo}%
\providecommand \bibfield  [0]{\@secondoftwo}%
\providecommand \translation [1]{[#1]}%
\providecommand \BibitemOpen [0]{}%
\providecommand \bibitemStop [0]{}%
\providecommand \bibitemNoStop [0]{.\EOS\space}%
\providecommand \EOS [0]{\spacefactor3000\relax}%
\providecommand \BibitemShut  [1]{\csname bibitem#1\endcsname}%
\let\auto@bib@innerbib\@empty
\bibitem [{\citenamefont {Horodecki}\ \emph {et~al.}(2009)\citenamefont {Horodecki}, \citenamefont {Horodecki}, \citenamefont {Horodecki},\ and\ \citenamefont {Horodecki}}]{Horodeckientanglementreview}%
  \BibitemOpen
  \bibfield  {author} {\bibinfo {author} {\bibfnamefont {R.}~\bibnamefont {Horodecki}}, \bibinfo {author} {\bibfnamefont {P.}~\bibnamefont {Horodecki}}, \bibinfo {author} {\bibfnamefont {M.}~\bibnamefont {Horodecki}},\ and\ \bibinfo {author} {\bibfnamefont {K.}~\bibnamefont {Horodecki}},\ }\bibfield  {title} {\bibinfo {title} {Quantum entanglement},\ }\href {https://doi.org/10.1103/RevModPhys.81.865} {\bibfield  {journal} {\bibinfo  {journal} {Rev. Mod. Phys.}\ }\textbf {\bibinfo {volume} {81}},\ \bibinfo {pages} {865} (\bibinfo {year} {2009})}\BibitemShut {NoStop}%
\bibitem [{\citenamefont {Bennett}\ \emph {et~al.}(1993)\citenamefont {Bennett}, \citenamefont {Brassard}, \citenamefont {Cr\'epeau}, \citenamefont {Jozsa}, \citenamefont {Peres},\ and\ \citenamefont {Wootters}}]{Benettteleportofstate}%
  \BibitemOpen
  \bibfield  {author} {\bibinfo {author} {\bibfnamefont {C.~H.}\ \bibnamefont {Bennett}}, \bibinfo {author} {\bibfnamefont {G.}~\bibnamefont {Brassard}}, \bibinfo {author} {\bibfnamefont {C.}~\bibnamefont {Cr\'epeau}}, \bibinfo {author} {\bibfnamefont {R.}~\bibnamefont {Jozsa}}, \bibinfo {author} {\bibfnamefont {A.}~\bibnamefont {Peres}},\ and\ \bibinfo {author} {\bibfnamefont {W.~K.}\ \bibnamefont {Wootters}},\ }\bibfield  {title} {\bibinfo {title} {Teleporting an unknown quantum state via dual classical and {Einstein}-{Podolsky}-{Rosen} channels},\ }\href {https://doi.org/10.1103/PhysRevLett.70.1895} {\bibfield  {journal} {\bibinfo  {journal} {Phys. Rev. Lett.}\ }\textbf {\bibinfo {volume} {70}},\ \bibinfo {pages} {1895} (\bibinfo {year} {1993})}\BibitemShut {NoStop}%
\bibitem [{\citenamefont {Bennett}\ and\ \citenamefont {Wiesner}(1992)}]{Benettsuperdensecoding}%
  \BibitemOpen
  \bibfield  {author} {\bibinfo {author} {\bibfnamefont {C.~H.}\ \bibnamefont {Bennett}}\ and\ \bibinfo {author} {\bibfnamefont {S.~J.}\ \bibnamefont {Wiesner}},\ }\bibfield  {title} {\bibinfo {title} {Communication via one- and two-particle operators on {Einstein}-{Podolsky}-{Rosen} states},\ }\href {https://doi.org/10.1103/PhysRevLett.69.2881} {\bibfield  {journal} {\bibinfo  {journal} {Phys. Rev. Lett.}\ }\textbf {\bibinfo {volume} {69}},\ \bibinfo {pages} {2881} (\bibinfo {year} {1992})}\BibitemShut {NoStop}%
\bibitem [{\citenamefont {Bennett}\ and\ \citenamefont {Brassard}(1984)}]{BennettQKD}%
  \BibitemOpen
  \bibfield  {author} {\bibinfo {author} {\bibfnamefont {C.~H.}\ \bibnamefont {Bennett}}\ and\ \bibinfo {author} {\bibfnamefont {G.}~\bibnamefont {Brassard}},\ }\bibfield  {title} {\bibinfo {title} {Quantum cryptography: Public key distribution and coin tossing},\ }in\ \href@noop {} {\emph {\bibinfo {booktitle} {Proceedings of IEEE International Conference on Computers, Systems and Signal Processing}}}\ (\bibinfo {address} {Bangalore, India},\ \bibinfo {year} {1984})\ pp.\ \bibinfo {pages} {175--179}\BibitemShut {NoStop}%
\bibitem [{\citenamefont {Pal}\ \emph {et~al.}(2014)\citenamefont {Pal}, \citenamefont {Bandyopadhyay},\ and\ \citenamefont {Ghosh}}]{Palnotmest}%
  \BibitemOpen
  \bibfield  {author} {\bibinfo {author} {\bibfnamefont {R.}~\bibnamefont {Pal}}, \bibinfo {author} {\bibfnamefont {S.}~\bibnamefont {Bandyopadhyay}},\ and\ \bibinfo {author} {\bibfnamefont {S.}~\bibnamefont {Ghosh}},\ }\bibfield  {title} {\bibinfo {title} {Entanglement sharing through noisy qubit channels: One-shot optimal singlet fraction},\ }\href {https://doi.org/10.1103/PhysRevA.90.052304} {\bibfield  {journal} {\bibinfo  {journal} {Phys. Rev. A}\ }\textbf {\bibinfo {volume} {90}},\ \bibinfo {pages} {052304} (\bibinfo {year} {2014})}\BibitemShut {NoStop}%
\bibitem [{\citenamefont {Streltsov}\ \emph {et~al.}(2015)\citenamefont {Streltsov}, \citenamefont {Augusiak}, \citenamefont {Demianowicz},\ and\ \citenamefont {Lewenstein}}]{Streltsovdistribution}%
  \BibitemOpen
  \bibfield  {author} {\bibinfo {author} {\bibfnamefont {A.}~\bibnamefont {Streltsov}}, \bibinfo {author} {\bibfnamefont {R.}~\bibnamefont {Augusiak}}, \bibinfo {author} {\bibfnamefont {M.}~\bibnamefont {Demianowicz}},\ and\ \bibinfo {author} {\bibfnamefont {M.}~\bibnamefont {Lewenstein}},\ }\bibfield  {title} {\bibinfo {title} {Progress towards a unified approach to entanglement distribution},\ }\href {https://doi.org/10.1103/PhysRevA.92.012335} {\bibfield  {journal} {\bibinfo  {journal} {Phys. Rev. A}\ }\textbf {\bibinfo {volume} {92}},\ \bibinfo {pages} {012335} (\bibinfo {year} {2015})}\BibitemShut {NoStop}%
\bibitem [{\citenamefont {Krisnanda}(2020)}]{Krisnandadistributionentanglement}%
  \BibitemOpen
  \bibfield  {author} {\bibinfo {author} {\bibfnamefont {T.}~\bibnamefont {Krisnanda}},\ }\bibfield  {title} {\bibinfo {title} {Distribution of quantum entanglement: Principles and applications},\ }\href {https://arxiv.org/abs/2003.08657} {\bibfield  {journal} {\bibinfo  {journal} {arXiv:2003.08657}\ } (\bibinfo {year} {2020})}\BibitemShut {NoStop}%
\bibitem [{\citenamefont {Zuppardo}\ \emph {et~al.}(2016)\citenamefont {Zuppardo}, \citenamefont {Krisnanda}, \citenamefont {Paterek}, \citenamefont {Bandyopadhyay}, \citenamefont {Banerjee}, \citenamefont {Deb}, \citenamefont {Halder}, \citenamefont {Modi},\ and\ \citenamefont {Paternostro}}]{Zuppardentanglementdistribution}%
  \BibitemOpen
  \bibfield  {author} {\bibinfo {author} {\bibfnamefont {M.}~\bibnamefont {Zuppardo}}, \bibinfo {author} {\bibfnamefont {T.}~\bibnamefont {Krisnanda}}, \bibinfo {author} {\bibfnamefont {T.}~\bibnamefont {Paterek}}, \bibinfo {author} {\bibfnamefont {S.}~\bibnamefont {Bandyopadhyay}}, \bibinfo {author} {\bibfnamefont {A.}~\bibnamefont {Banerjee}}, \bibinfo {author} {\bibfnamefont {P.}~\bibnamefont {Deb}}, \bibinfo {author} {\bibfnamefont {S.}~\bibnamefont {Halder}}, \bibinfo {author} {\bibfnamefont {K.}~\bibnamefont {Modi}},\ and\ \bibinfo {author} {\bibfnamefont {M.}~\bibnamefont {Paternostro}},\ }\bibfield  {title} {\bibinfo {title} {Excessive distribution of quantum entanglement},\ }\href {https://doi.org/10.1103/PhysRevA.93.012305} {\bibfield  {journal} {\bibinfo  {journal} {Phys. Rev. A}\ }\textbf {\bibinfo {volume} {93}},\ \bibinfo {pages} {012305} (\bibinfo {year} {2016})}\BibitemShut {NoStop}%
\bibitem [{\citenamefont {Streltsov}\ \emph {et~al.}(2012)\citenamefont {Streltsov}, \citenamefont {Kampermann},\ and\ \citenamefont {Bru\ss{}}}]{Streltsovdistributionentangpre}%
  \BibitemOpen
  \bibfield  {author} {\bibinfo {author} {\bibfnamefont {A.}~\bibnamefont {Streltsov}}, \bibinfo {author} {\bibfnamefont {H.}~\bibnamefont {Kampermann}},\ and\ \bibinfo {author} {\bibfnamefont {D.}~\bibnamefont {Bru\ss{}}},\ }\bibfield  {title} {\bibinfo {title} {Quantum cost for sending entanglement},\ }\href {https://doi.org/10.1103/PhysRevLett.108.250501} {\bibfield  {journal} {\bibinfo  {journal} {Phys. Rev. Lett.}\ }\textbf {\bibinfo {volume} {108}},\ \bibinfo {pages} {250501} (\bibinfo {year} {2012})}\BibitemShut {NoStop}%
\bibitem [{\citenamefont {Chuan}\ \emph {et~al.}(2012)\citenamefont {Chuan}, \citenamefont {Maillard}, \citenamefont {Modi}, \citenamefont {Paterek}, \citenamefont {Paternostro},\ and\ \citenamefont {Piani}}]{Chuand}%
  \BibitemOpen
  \bibfield  {author} {\bibinfo {author} {\bibfnamefont {T.~K.}\ \bibnamefont {Chuan}}, \bibinfo {author} {\bibfnamefont {J.}~\bibnamefont {Maillard}}, \bibinfo {author} {\bibfnamefont {K.}~\bibnamefont {Modi}}, \bibinfo {author} {\bibfnamefont {T.}~\bibnamefont {Paterek}}, \bibinfo {author} {\bibfnamefont {M.}~\bibnamefont {Paternostro}},\ and\ \bibinfo {author} {\bibfnamefont {M.}~\bibnamefont {Piani}},\ }\bibfield  {title} {\bibinfo {title} {Quantum discord bounds the amount of distributed entanglement},\ }\href {https://doi.org/10.1103/PhysRevLett.109.070501} {\bibfield  {journal} {\bibinfo  {journal} {Phys. Rev. Lett.}\ }\textbf {\bibinfo {volume} {109}},\ \bibinfo {pages} {070501} (\bibinfo {year} {2012})}\BibitemShut {NoStop}%
\bibitem [{\citenamefont {Gurvits}(2003)}]{Gurvitsnph}%
  \BibitemOpen
  \bibfield  {author} {\bibinfo {author} {\bibfnamefont {L.}~\bibnamefont {Gurvits}},\ }\bibfield  {title} {\bibinfo {title} {Classical deterministic complexity of edmonds' problem and quantum entanglement},\ }in\ \href {https://doi.org/10.1145/780542.780545} {\emph {\bibinfo {booktitle} {Proceedings of the Thirty-Fifth Annual ACM Symposium on Theory of Computing}}},\ \bibinfo {series and number} {STOC '03}\ (\bibinfo  {publisher} {Association for Computing Machinery},\ \bibinfo {address} {New York, NY, USA},\ \bibinfo {year} {2003})\ p.\ \bibinfo {pages} {10–19}\BibitemShut {NoStop}%
\bibitem [{\citenamefont {Peres}(1996)}]{Perespptcriterion}%
  \BibitemOpen
  \bibfield  {author} {\bibinfo {author} {\bibfnamefont {A.}~\bibnamefont {Peres}},\ }\bibfield  {title} {\bibinfo {title} {Separability criterion for density matrices},\ }\href {https://doi.org/10.1103/PhysRevLett.77.1413} {\bibfield  {journal} {\bibinfo  {journal} {Phys. Rev. Lett.}\ }\textbf {\bibinfo {volume} {77}},\ \bibinfo {pages} {1413} (\bibinfo {year} {1996})}\BibitemShut {NoStop}%
\bibitem [{\citenamefont {Horodecki}\ \emph {et~al.}(1996)\citenamefont {Horodecki}, \citenamefont {Horodecki},\ and\ \citenamefont {Horodecki}}]{Horodeckipptcriterion}%
  \BibitemOpen
  \bibfield  {author} {\bibinfo {author} {\bibfnamefont {M.}~\bibnamefont {Horodecki}}, \bibinfo {author} {\bibfnamefont {P.}~\bibnamefont {Horodecki}},\ and\ \bibinfo {author} {\bibfnamefont {R.}~\bibnamefont {Horodecki}},\ }\bibfield  {title} {\bibinfo {title} {Separability of mixed states: necessary and sufficient conditions},\ }\href {https://doi.org/https://doi.org/10.1016/S0375-9601(96)00706-2} {\bibfield  {journal} {\bibinfo  {journal} {Physics Letters A}\ }\textbf {\bibinfo {volume} {223}},\ \bibinfo {pages} {1} (\bibinfo {year} {1996})}\BibitemShut {NoStop}%
\bibitem [{\citenamefont {Horodecki}(1997)}]{HORODECKI1997333}%
  \BibitemOpen
  \bibfield  {author} {\bibinfo {author} {\bibfnamefont {P.}~\bibnamefont {Horodecki}},\ }\bibfield  {title} {\bibinfo {title} {Separability criterion and inseparable mixed states with positive partial transposition},\ }\href {https://doi.org/https://doi.org/10.1016/S0375-9601(97)00416-7} {\bibfield  {journal} {\bibinfo  {journal} {Physics Letters A}\ }\textbf {\bibinfo {volume} {232}},\ \bibinfo {pages} {333} (\bibinfo {year} {1997})}\BibitemShut {NoStop}%
\bibitem [{\citenamefont {Horodecki}\ \emph {et~al.}(1998)\citenamefont {Horodecki}, \citenamefont {Horodecki},\ and\ \citenamefont {Horodecki}}]{Horodeckipptbound}%
  \BibitemOpen
  \bibfield  {author} {\bibinfo {author} {\bibfnamefont {M.}~\bibnamefont {Horodecki}}, \bibinfo {author} {\bibfnamefont {P.}~\bibnamefont {Horodecki}},\ and\ \bibinfo {author} {\bibfnamefont {R.}~\bibnamefont {Horodecki}},\ }\bibfield  {title} {\bibinfo {title} {Mixed-state entanglement and distillation: Is there a ``bound'' entanglement in nature?},\ }\href {https://doi.org/10.1103/PhysRevLett.80.5239} {\bibfield  {journal} {\bibinfo  {journal} {Phys. Rev. Lett.}\ }\textbf {\bibinfo {volume} {80}},\ \bibinfo {pages} {5239} (\bibinfo {year} {1998})}\BibitemShut {NoStop}%
\bibitem [{\citenamefont {\ifmmode~\dot{Z}\else \.{Z}\fi{}yczkowski}\ \emph {et~al.}(1998)\citenamefont {\ifmmode~\dot{Z}\else \.{Z}\fi{}yczkowski}, \citenamefont {Horodecki}, \citenamefont {Sanpera},\ and\ \citenamefont {Lewenstein}}]{PhysRevA.58.883}%
  \BibitemOpen
  \bibfield  {author} {\bibinfo {author} {\bibfnamefont {K.}~\bibnamefont {\ifmmode~\dot{Z}\else \.{Z}\fi{}yczkowski}}, \bibinfo {author} {\bibfnamefont {P.}~\bibnamefont {Horodecki}}, \bibinfo {author} {\bibfnamefont {A.}~\bibnamefont {Sanpera}},\ and\ \bibinfo {author} {\bibfnamefont {M.}~\bibnamefont {Lewenstein}},\ }\bibfield  {title} {\bibinfo {title} {Volume of the set of separable states},\ }\href {https://doi.org/10.1103/PhysRevA.58.883} {\bibfield  {journal} {\bibinfo  {journal} {Phys. Rev. A}\ }\textbf {\bibinfo {volume} {58}},\ \bibinfo {pages} {883} (\bibinfo {year} {1998})}\BibitemShut {NoStop}%
\bibitem [{\citenamefont {Vidal}\ and\ \citenamefont {Werner}(2002)}]{Vidalnegat}%
  \BibitemOpen
  \bibfield  {author} {\bibinfo {author} {\bibfnamefont {G.}~\bibnamefont {Vidal}}\ and\ \bibinfo {author} {\bibfnamefont {R.~F.}\ \bibnamefont {Werner}},\ }\bibfield  {title} {\bibinfo {title} {Computable measure of entanglement},\ }\href {https://doi.org/10.1103/PhysRevA.65.032314} {\bibfield  {journal} {\bibinfo  {journal} {Phys. Rev. A}\ }\textbf {\bibinfo {volume} {65}},\ \bibinfo {pages} {032314} (\bibinfo {year} {2002})}\BibitemShut {NoStop}%
\bibitem [{\citenamefont {Sanpera}\ \emph {et~al.}(1998)\citenamefont {Sanpera}, \citenamefont {Tarrach},\ and\ \citenamefont {Vidal}}]{Sanperaonenegativeeigenvalue}%
  \BibitemOpen
  \bibfield  {author} {\bibinfo {author} {\bibfnamefont {A.}~\bibnamefont {Sanpera}}, \bibinfo {author} {\bibfnamefont {R.}~\bibnamefont {Tarrach}},\ and\ \bibinfo {author} {\bibfnamefont {G.}~\bibnamefont {Vidal}},\ }\bibfield  {title} {\bibinfo {title} {Local description of quantum inseparability},\ }\href {https://doi.org/10.1103/PhysRevA.58.826} {\bibfield  {journal} {\bibinfo  {journal} {Phys. Rev. A}\ }\textbf {\bibinfo {volume} {58}},\ \bibinfo {pages} {826} (\bibinfo {year} {1998})}\BibitemShut {NoStop}%
\bibitem [{\citenamefont {Skrzypczyk}\ and\ \citenamefont {Cavalcanti}(2023)}]{Skrzypczykquantumsdp}%
  \BibitemOpen
  \bibfield  {author} {\bibinfo {author} {\bibfnamefont {P.}~\bibnamefont {Skrzypczyk}}\ and\ \bibinfo {author} {\bibfnamefont {D.}~\bibnamefont {Cavalcanti}},\ }\href {https://doi.org/10.1088/978-0-7503-3343-6} {\emph {\bibinfo {title} {Semidefinite Programming in Quantum Information Science}}},\ 2053-2563\ (\bibinfo  {publisher} {IOP Publishing},\ \bibinfo {year} {2023})\BibitemShut {NoStop}%
\bibitem [{\citenamefont {Diamond}\ and\ \citenamefont {Boyd}(2016)}]{diamond2016cvxpy}%
  \BibitemOpen
  \bibfield  {author} {\bibinfo {author} {\bibfnamefont {S.}~\bibnamefont {Diamond}}\ and\ \bibinfo {author} {\bibfnamefont {S.}~\bibnamefont {Boyd}},\ }\bibfield  {title} {\bibinfo {title} {{CVXPY}: {A} {P}ython-embedded modeling language for convex optimization},\ }\href {https://doi.org/10.48550/arXiv.1603.00943} {\bibfield  {journal} {\bibinfo  {journal} {Journal of Machine Learning Research}\ }\textbf {\bibinfo {volume} {17}},\ \bibinfo {pages} {1} (\bibinfo {year} {2016})}\BibitemShut {NoStop}%
\bibitem [{\citenamefont {Agrawal}\ \emph {et~al.}(2018)\citenamefont {Agrawal}, \citenamefont {Verschueren}, \citenamefont {Diamond},\ and\ \citenamefont {Boyd}}]{agrawal2018rewriting}%
  \BibitemOpen
  \bibfield  {author} {\bibinfo {author} {\bibfnamefont {A.}~\bibnamefont {Agrawal}}, \bibinfo {author} {\bibfnamefont {R.}~\bibnamefont {Verschueren}}, \bibinfo {author} {\bibfnamefont {S.}~\bibnamefont {Diamond}},\ and\ \bibinfo {author} {\bibfnamefont {S.}~\bibnamefont {Boyd}},\ }\bibfield  {title} {\bibinfo {title} {A rewriting system for convex optimization problems},\ }\href {https://doi.org/10.48550/arXiv.1709.04494} {\bibfield  {journal} {\bibinfo  {journal} {Journal of Control and Decision}\ }\textbf {\bibinfo {volume} {5}},\ \bibinfo {pages} {42} (\bibinfo {year} {2018})}\BibitemShut {NoStop}%
\bibitem [{\citenamefont {Johansson}\ \emph {et~al.}(2013)\citenamefont {Johansson}, \citenamefont {Nation},\ and\ \citenamefont {Nori}}]{JOHANSSONqutippython}%
  \BibitemOpen
  \bibfield  {author} {\bibinfo {author} {\bibfnamefont {J.}~\bibnamefont {Johansson}}, \bibinfo {author} {\bibfnamefont {P.}~\bibnamefont {Nation}},\ and\ \bibinfo {author} {\bibfnamefont {F.}~\bibnamefont {Nori}},\ }\bibfield  {title} {\bibinfo {title} {Qutip 2: A python framework for the dynamics of open quantum systems},\ }\href {https://doi.org/https://doi.org/10.1016/j.cpc.2012.11.019} {\bibfield  {journal} {\bibinfo  {journal} {Computer Physics Communications}\ }\textbf {\bibinfo {volume} {184}},\ \bibinfo {pages} {1234} (\bibinfo {year} {2013})}\BibitemShut {NoStop}%
\bibitem [{\citenamefont {Harris}\ \emph {et~al.}(2020)\citenamefont {Harris}, \citenamefont {Millman}, \citenamefont {van~der Walt}, \citenamefont {Gommers}, \citenamefont {Virtanen}, \citenamefont {Cournapeau}, \citenamefont {Wieser}, \citenamefont {Taylor}, \citenamefont {Berg}, \citenamefont {Smith}, \citenamefont {Kern}, \citenamefont {Picus}, \citenamefont {Hoyer}, \citenamefont {van Kerkwijk}, \citenamefont {Brett}, \citenamefont {Haldane}, \citenamefont {del R{\'{i}}o}, \citenamefont {Wiebe}, \citenamefont {Peterson}, \citenamefont {G{\'{e}}rard-Marchant}, \citenamefont {Sheppard}, \citenamefont {Reddy}, \citenamefont {Weckesser}, \citenamefont {Abbasi}, \citenamefont {Gohlke},\ and\ \citenamefont {Oliphant}}]{Harrisnumpy}%
  \BibitemOpen
  \bibfield  {author} {\bibinfo {author} {\bibfnamefont {C.~R.}\ \bibnamefont {Harris}}, \bibinfo {author} {\bibfnamefont {K.~J.}\ \bibnamefont {Millman}}, \bibinfo {author} {\bibfnamefont {S.~J.}\ \bibnamefont {van~der Walt}}, \bibinfo {author} {\bibfnamefont {R.}~\bibnamefont {Gommers}}, \bibinfo {author} {\bibfnamefont {P.}~\bibnamefont {Virtanen}}, \bibinfo {author} {\bibfnamefont {D.}~\bibnamefont {Cournapeau}}, \bibinfo {author} {\bibfnamefont {E.}~\bibnamefont {Wieser}}, \bibinfo {author} {\bibfnamefont {J.}~\bibnamefont {Taylor}}, \bibinfo {author} {\bibfnamefont {S.}~\bibnamefont {Berg}}, \bibinfo {author} {\bibfnamefont {N.~J.}\ \bibnamefont {Smith}}, \bibinfo {author} {\bibfnamefont {R.}~\bibnamefont {Kern}}, \bibinfo {author} {\bibfnamefont {M.}~\bibnamefont {Picus}}, \bibinfo {author} {\bibfnamefont {S.}~\bibnamefont {Hoyer}}, \bibinfo {author} {\bibfnamefont {M.~H.}\ \bibnamefont {van Kerkwijk}}, \bibinfo {author} {\bibfnamefont {M.}~\bibnamefont {Brett}}, \bibinfo {author} {\bibfnamefont
  {A.}~\bibnamefont {Haldane}}, \bibinfo {author} {\bibfnamefont {J.~F.}\ \bibnamefont {del R{\'{i}}o}}, \bibinfo {author} {\bibfnamefont {M.}~\bibnamefont {Wiebe}}, \bibinfo {author} {\bibfnamefont {P.}~\bibnamefont {Peterson}}, \bibinfo {author} {\bibfnamefont {P.}~\bibnamefont {G{\'{e}}rard-Marchant}}, \bibinfo {author} {\bibfnamefont {K.}~\bibnamefont {Sheppard}}, \bibinfo {author} {\bibfnamefont {T.}~\bibnamefont {Reddy}}, \bibinfo {author} {\bibfnamefont {W.}~\bibnamefont {Weckesser}}, \bibinfo {author} {\bibfnamefont {H.}~\bibnamefont {Abbasi}}, \bibinfo {author} {\bibfnamefont {C.}~\bibnamefont {Gohlke}},\ and\ \bibinfo {author} {\bibfnamefont {T.~E.}\ \bibnamefont {Oliphant}},\ }\bibfield  {title} {\bibinfo {title} {Array programming with {NumPy}},\ }\href {https://doi.org/10.1038/s41586-020-2649-2} {\bibfield  {journal} {\bibinfo  {journal} {Nature}\ }\textbf {\bibinfo {volume} {585}},\ \bibinfo {pages} {357} (\bibinfo {year} {2020})}\BibitemShut {NoStop}%
\bibitem [{\citenamefont {Nielsen}\ and\ \citenamefont {Chuang}(2010)}]{Nielsenquantumcomputation}%
  \BibitemOpen
  \bibfield  {author} {\bibinfo {author} {\bibfnamefont {M.~A.}\ \bibnamefont {Nielsen}}\ and\ \bibinfo {author} {\bibfnamefont {I.~L.}\ \bibnamefont {Chuang}},\ }\href {https://doi.org/https://doi.org/10.1017/CBO9780511976667} {\emph {\bibinfo {title} {Quantum Computation and Quantum Information: 10th Anniversary Edition}}}\ (\bibinfo  {publisher} {Cambridge University Press},\ \bibinfo {year} {2010})\BibitemShut {NoStop}%
\bibitem [{\citenamefont {Choi}(1975)}]{CHOIch}%
  \BibitemOpen
  \bibfield  {author} {\bibinfo {author} {\bibfnamefont {M.-D.}\ \bibnamefont {Choi}},\ }\bibfield  {title} {\bibinfo {title} {Completely positive linear maps on complex matrices},\ }\href {https://doi.org/https://doi.org/10.1016/0024-3795(75)90075-0} {\bibfield  {journal} {\bibinfo  {journal} {Linear Algebra and its Applications}\ }\textbf {\bibinfo {volume} {10}},\ \bibinfo {pages} {285} (\bibinfo {year} {1975})}\BibitemShut {NoStop}%
\bibitem [{\citenamefont {Gyongyosi}\ and\ \citenamefont {Imre}(2012)}]{gyongyosipropertieschannel}%
  \BibitemOpen
  \bibfield  {author} {\bibinfo {author} {\bibfnamefont {L.}~\bibnamefont {Gyongyosi}}\ and\ \bibinfo {author} {\bibfnamefont {S.}~\bibnamefont {Imre}},\ }\href {https://arxiv.org/abs/1208.1270} {\bibinfo {title} {Properties of the quantum channel}} (\bibinfo {year} {2012})\BibitemShut {NoStop}%
\bibitem [{\citenamefont {Filippov}\ \emph {et~al.}(2012)\citenamefont {Filippov}, \citenamefont {Ryb\'ar},\ and\ \citenamefont {Ziman}}]{Filippovea}%
  \BibitemOpen
  \bibfield  {author} {\bibinfo {author} {\bibfnamefont {S.~N.}\ \bibnamefont {Filippov}}, \bibinfo {author} {\bibfnamefont {T.}~\bibnamefont {Ryb\'ar}},\ and\ \bibinfo {author} {\bibfnamefont {M.}~\bibnamefont {Ziman}},\ }\bibfield  {title} {\bibinfo {title} {Local two-qubit entanglement-annihilating channels},\ }\href {https://doi.org/10.1103/PhysRevA.85.012303} {\bibfield  {journal} {\bibinfo  {journal} {Phys. Rev. A}\ }\textbf {\bibinfo {volume} {85}},\ \bibinfo {pages} {012303} (\bibinfo {year} {2012})}\BibitemShut {NoStop}%
\bibitem [{\citenamefont {Moravcikova}\ and\ \citenamefont {Ziman}(2010)}]{MoravcikovaEAandEBchannels}%
  \BibitemOpen
  \bibfield  {author} {\bibinfo {author} {\bibfnamefont {L.}~\bibnamefont {Moravcikova}}\ and\ \bibinfo {author} {\bibfnamefont {M.}~\bibnamefont {Ziman}},\ }\bibfield  {title} {\bibinfo {title} {Entanglement-annihilating and entanglement-breaking channels},\ }\href {https://doi.org/10.1088/1751-8113/43/27/275306} {\bibfield  {journal} {\bibinfo  {journal} {Journal of Physics A: Mathematical and Theoretical}\ }\textbf {\bibinfo {volume} {43}},\ \bibinfo {pages} {275306} (\bibinfo {year} {2010})}\BibitemShut {NoStop}%
\bibitem [{\citenamefont {Filippov}\ and\ \citenamefont {Ziman}(2013)}]{Filippovbipartiteeachannels}%
  \BibitemOpen
  \bibfield  {author} {\bibinfo {author} {\bibfnamefont {S.~N.}\ \bibnamefont {Filippov}}\ and\ \bibinfo {author} {\bibfnamefont {M.}~\bibnamefont {Ziman}},\ }\bibfield  {title} {\bibinfo {title} {Bipartite entanglement-annihilating maps: Necessary and sufficient conditions},\ }\href {https://doi.org/10.1103/PhysRevA.88.032316} {\bibfield  {journal} {\bibinfo  {journal} {Phys. Rev. A}\ }\textbf {\bibinfo {volume} {88}},\ \bibinfo {pages} {032316} (\bibinfo {year} {2013})}\BibitemShut {NoStop}%
\bibitem [{\citenamefont {Horodecki}\ \emph {et~al.}(2003)\citenamefont {Horodecki}, \citenamefont {Shor},\ and\ \citenamefont {Ruskai}}]{horodecki2003entanglement}%
  \BibitemOpen
  \bibfield  {author} {\bibinfo {author} {\bibfnamefont {M.}~\bibnamefont {Horodecki}}, \bibinfo {author} {\bibfnamefont {P.~W.}\ \bibnamefont {Shor}},\ and\ \bibinfo {author} {\bibfnamefont {M.~B.}\ \bibnamefont {Ruskai}},\ }\bibfield  {title} {\bibinfo {title} {Entanglement breaking channels},\ }\href {https://doi.org/https://doi.org/10.1142/S0129055X03001709} {\bibfield  {journal} {\bibinfo  {journal} {Reviews in Mathematical Physics}\ }\textbf {\bibinfo {volume} {15}},\ \bibinfo {pages} {629} (\bibinfo {year} {2003})}\BibitemShut {NoStop}%
\bibitem [{\citenamefont {Ruskai}(2003)}]{RuskaiEB}%
  \BibitemOpen
  \bibfield  {author} {\bibinfo {author} {\bibfnamefont {M.~B.}\ \bibnamefont {Ruskai}},\ }\bibfield  {title} {\bibinfo {title} {Qubit entanglement breaking channels},\ }\href {https://doi.org/10.1142/S0129055X03001710} {\bibfield  {journal} {\bibinfo  {journal} {Reviews in Mathematical Physics}\ }\textbf {\bibinfo {volume} {15}},\ \bibinfo {pages} {643} (\bibinfo {year} {2003})}\BibitemShut {NoStop}%
\bibitem [{\citenamefont {Holevo}\ \emph {et~al.}(2005)\citenamefont {Holevo}, \citenamefont {Shirokov},\ and\ \citenamefont {Werner}}]{HolevoEBIN}%
  \BibitemOpen
  \bibfield  {author} {\bibinfo {author} {\bibfnamefont {A.~S.}\ \bibnamefont {Holevo}}, \bibinfo {author} {\bibfnamefont {M.~E.}\ \bibnamefont {Shirokov}},\ and\ \bibinfo {author} {\bibfnamefont {R.~F.}\ \bibnamefont {Werner}},\ }\href {https://arxiv.org/abs/quant-ph/0504204} {\bibinfo {title} {Separability and entanglement-breaking in infinite dimensions}} (\bibinfo {year} {2005})\BibitemShut {NoStop}%
\bibitem [{\citenamefont {Li}\ and\ \citenamefont {Choi}(2023)}]{li2023unital}%
  \BibitemOpen
  \bibfield  {author} {\bibinfo {author} {\bibfnamefont {C.-K.}\ \bibnamefont {Li}}\ and\ \bibinfo {author} {\bibfnamefont {M.-D.}\ \bibnamefont {Choi}},\ }\bibfield  {title} {\bibinfo {title} {On unital qubit channels},\ }\href {https://doi.org/10.48550/arXiv.2301.01358} {\bibfield  {journal} {\bibinfo  {journal} {arXiv:2301.01358}\ } (\bibinfo {year} {2023})}\BibitemShut {NoStop}%
\bibitem [{\citenamefont {Khatri}\ and\ \citenamefont {Wilde}(2024)}]{KhatriPQCT}%
  \BibitemOpen
  \bibfield  {author} {\bibinfo {author} {\bibfnamefont {S.}~\bibnamefont {Khatri}}\ and\ \bibinfo {author} {\bibfnamefont {M.~M.}\ \bibnamefont {Wilde}},\ }\href {https://arxiv.org/abs/2011.04672} {\bibinfo {title} {Principles of quantum communication theory: A modern approach}} (\bibinfo {year} {2024})\BibitemShut {NoStop}%
\bibitem [{\citenamefont {Wei}\ and\ \citenamefont {Goldbart}(2003)}]{Weigeometricentanglement}%
  \BibitemOpen
  \bibfield  {author} {\bibinfo {author} {\bibfnamefont {T.-C.}\ \bibnamefont {Wei}}\ and\ \bibinfo {author} {\bibfnamefont {P.~M.}\ \bibnamefont {Goldbart}},\ }\bibfield  {title} {\bibinfo {title} {Geometric measure of entanglement and applications to bipartite and multipartite quantum states},\ }\href {https://doi.org/10.1103/PhysRevA.68.042307} {\bibfield  {journal} {\bibinfo  {journal} {Phys. Rev. A}\ }\textbf {\bibinfo {volume} {68}},\ \bibinfo {pages} {042307} (\bibinfo {year} {2003})}\BibitemShut {NoStop}%
\bibitem [{\citenamefont {Augusiak}\ \emph {et~al.}(2008)\citenamefont {Augusiak}, \citenamefont {Demianowicz},\ and\ \citenamefont {Horodecki}}]{Augusiakdeterminant}%
  \BibitemOpen
  \bibfield  {author} {\bibinfo {author} {\bibfnamefont {R.}~\bibnamefont {Augusiak}}, \bibinfo {author} {\bibfnamefont {M.}~\bibnamefont {Demianowicz}},\ and\ \bibinfo {author} {\bibfnamefont {P.}~\bibnamefont {Horodecki}},\ }\bibfield  {title} {\bibinfo {title} {Universal observable detecting all two-qubit entanglement and determinant-based separability tests},\ }\href {https://doi.org/10.1103/PhysRevA.77.030301} {\bibfield  {journal} {\bibinfo  {journal} {Phys. Rev. A}\ }\textbf {\bibinfo {volume} {77}},\ \bibinfo {pages} {030301} (\bibinfo {year} {2008})}\BibitemShut {NoStop}%
\bibitem [{\citenamefont {Ganardi}\ \emph {et~al.}(2025)\citenamefont {Ganardi}, \citenamefont {Masajada}, \citenamefont {Naseri},\ and\ \citenamefont {Streltsov}}]{Ganardilocalpurity}%
  \BibitemOpen
  \bibfield  {author} {\bibinfo {author} {\bibfnamefont {R.}~\bibnamefont {Ganardi}}, \bibinfo {author} {\bibfnamefont {P.}~\bibnamefont {Masajada}}, \bibinfo {author} {\bibfnamefont {M.}~\bibnamefont {Naseri}},\ and\ \bibinfo {author} {\bibfnamefont {A.}~\bibnamefont {Streltsov}},\ }\bibfield  {title} {\bibinfo {title} {Local {P}urity {D}istillation in {Q}uantum {S}ystems: {E}xploring the {C}omplementarity {B}etween {P}urity and {E}ntanglement},\ }\href {https://doi.org/10.22331/q-2025-03-20-1666} {\bibfield  {journal} {\bibinfo  {journal} {{Quantum}}\ }\textbf {\bibinfo {volume} {9}},\ \bibinfo {pages} {1666} (\bibinfo {year} {2025})}\BibitemShut {NoStop}%
\bibitem [{\citenamefont {Chirolli}\ and\ \citenamefont {Burkard}(2008)}]{ChirolliGADquantumcomputing}%
  \BibitemOpen
  \bibfield  {author} {\bibinfo {author} {\bibfnamefont {L.}~\bibnamefont {Chirolli}}\ and\ \bibinfo {author} {\bibfnamefont {G.}~\bibnamefont {Burkard}},\ }\bibfield  {title} {\bibinfo {title} {Decoherence in solid-state qubits},\ }\href {https://doi.org/10.1080/00018730802218067} {\bibfield  {journal} {\bibinfo  {journal} {Advances in Physics}\ }\textbf {\bibinfo {volume} {57}},\ \bibinfo {pages} {225} (\bibinfo {year} {2008})}\BibitemShut {NoStop}%
\bibitem [{\citenamefont {Zou}\ \emph {et~al.}(2017)\citenamefont {Zou}, \citenamefont {Li}, \citenamefont {Wang}, \citenamefont {Cao}, \citenamefont {Ren}, \citenamefont {Yin}, \citenamefont {Peng}, \citenamefont {Wang},\ and\ \citenamefont {Pan}}]{ZouGADliop}%
  \BibitemOpen
  \bibfield  {author} {\bibinfo {author} {\bibfnamefont {W.-J.}\ \bibnamefont {Zou}}, \bibinfo {author} {\bibfnamefont {Y.-H.}\ \bibnamefont {Li}}, \bibinfo {author} {\bibfnamefont {S.-C.}\ \bibnamefont {Wang}}, \bibinfo {author} {\bibfnamefont {Y.}~\bibnamefont {Cao}}, \bibinfo {author} {\bibfnamefont {J.-G.}\ \bibnamefont {Ren}}, \bibinfo {author} {\bibfnamefont {J.}~\bibnamefont {Yin}}, \bibinfo {author} {\bibfnamefont {C.-Z.}\ \bibnamefont {Peng}}, \bibinfo {author} {\bibfnamefont {X.-B.}\ \bibnamefont {Wang}},\ and\ \bibinfo {author} {\bibfnamefont {J.-W.}\ \bibnamefont {Pan}},\ }\bibfield  {title} {\bibinfo {title} {Protecting entanglement from finite-temperature thermal noise via weak measurement and quantum measurement reversal},\ }\href {https://doi.org/10.1103/PhysRevA.95.042342} {\bibfield  {journal} {\bibinfo  {journal} {Phys. Rev. A}\ }\textbf {\bibinfo {volume} {95}},\ \bibinfo {pages} {042342} (\bibinfo {year} {2017})}\BibitemShut {NoStop}%
\bibitem [{\citenamefont {Khatri}\ \emph {et~al.}(2020)\citenamefont {Khatri}, \citenamefont {Sharma},\ and\ \citenamefont {Wilde}}]{KhartiGADchannel}%
  \BibitemOpen
  \bibfield  {author} {\bibinfo {author} {\bibfnamefont {S.}~\bibnamefont {Khatri}}, \bibinfo {author} {\bibfnamefont {K.}~\bibnamefont {Sharma}},\ and\ \bibinfo {author} {\bibfnamefont {M.~M.}\ \bibnamefont {Wilde}},\ }\bibfield  {title} {\bibinfo {title} {Information-theoretic aspects of the generalized amplitude-damping channel},\ }\href {https://doi.org/10.1103/PhysRevA.102.012401} {\bibfield  {journal} {\bibinfo  {journal} {Phys. Rev. A}\ }\textbf {\bibinfo {volume} {102}},\ \bibinfo {pages} {012401} (\bibinfo {year} {2020})}\BibitemShut {NoStop}%
\bibitem [{\citenamefont {Myatt}\ \emph {et~al.}(2000)\citenamefont {Myatt}, \citenamefont {King}, \citenamefont {Turchette}, \citenamefont {Sackett}, \citenamefont {Kielpinski}, \citenamefont {Itano}, \citenamefont {Monroe},\ and\ \citenamefont {Wineland}}]{MyattGAD}%
  \BibitemOpen
  \bibfield  {author} {\bibinfo {author} {\bibfnamefont {C.~J.}\ \bibnamefont {Myatt}}, \bibinfo {author} {\bibfnamefont {B.~E.}\ \bibnamefont {King}}, \bibinfo {author} {\bibfnamefont {Q.~A.}\ \bibnamefont {Turchette}}, \bibinfo {author} {\bibfnamefont {C.~A.}\ \bibnamefont {Sackett}}, \bibinfo {author} {\bibfnamefont {D.}~\bibnamefont {Kielpinski}}, \bibinfo {author} {\bibfnamefont {W.~M.}\ \bibnamefont {Itano}}, \bibinfo {author} {\bibfnamefont {C.}~\bibnamefont {Monroe}},\ and\ \bibinfo {author} {\bibfnamefont {D.~J.}\ \bibnamefont {Wineland}},\ }\bibfield  {title} {\bibinfo {title} {Decoherence of quantum superpositions through coupling to engineered reservoirs},\ }\href {https://doi.org/10.1038/35002001} {\bibfield  {journal} {\bibinfo  {journal} {Nature}\ }\textbf {\bibinfo {volume} {403}},\ \bibinfo {pages} {269} (\bibinfo {year} {2000})}\BibitemShut {NoStop}%
\bibitem [{\citenamefont {Lami}\ and\ \citenamefont {Giovannetti}(2015)}]{LamiGAD}%
  \BibitemOpen
  \bibfield  {author} {\bibinfo {author} {\bibfnamefont {L.}~\bibnamefont {Lami}}\ and\ \bibinfo {author} {\bibfnamefont {V.}~\bibnamefont {Giovannetti}},\ }\bibfield  {title} {\bibinfo {title} {Entanglement–breaking indices},\ }\href {https://doi.org/10.1063/1.4931482} {\bibfield  {journal} {\bibinfo  {journal} {Journal of Mathematical Physics}\ }\textbf {\bibinfo {volume} {56}},\ \bibinfo {pages} {092201} (\bibinfo {year} {2015})}\BibitemShut {NoStop}%
\end{thebibliography}%
\end{document}